\documentclass[12pt,article,onecolumn,textrm]{article}
\usepackage[english]{babel}
\usepackage[T1]{fontenc}
\usepackage{lmodern}
\usepackage{appendix}
\usepackage{amsmath}
\usepackage{mathrsfs}
\usepackage{amsfonts}
\usepackage{amssymb}
\usepackage{amsthm}
\usepackage{bbm}
\usepackage{color}
\usepackage{xcolor}[cmyk]
\usepackage{footmisc}
\usepackage[hyperindex,breaklinks]{hyperref}
\hypersetup{colorlinks   = true, 
	urlcolor     = black, 
	linkcolor    = blue, 
	citecolor   = teal}

\usepackage{cleveref}
\Crefname{equation}{Eq.}{Eqs.}
\Crefname{figure}{Fig.}{Figs.}
\Crefname{tabular}{Tab.}{Tabs.}

\usepackage{graphicx}
\usepackage{enumerate}
\usepackage{cases}
\usepackage{authblk}
\usepackage{tikz-cd} 
\usepackage{ytableau}

\newcommand\mlnode[1]{\fbox{\begin{tabular}{@{}c@{}}#1\end{tabular}}}

\usepackage[sort&compress,numbers]{natbib}

\theoremstyle{plain}

\newtheorem{thm}{Theorem}[section] 
\newtheorem{lem}[thm]{Lemma}
\newtheorem{prop}[thm]{Proposition}
\newtheorem{cor}[thm]{Corollary}
\newtheorem{defi}[thm]{Definition}
\newtheorem*{defi*}{Definition}
\newtheorem{example}[thm]{Example}

\newtheorem*{thm*}{Theorem}

\theoremstyle{definition}
\newtheorem*{remark}{Remark}

\newcommand{\nn}{\mathbb{N}}
\newcommand{\zz}{\mathbb{Z}}
\newcommand{\rr}{\mathbb{R}}
\newcommand{\cc}{\mathbb{C}}
\newcommand{\Abb}{\mathbb{A}}

\newcommand{\dd}{\mathrm{\normalfont d}}

\newcommand{\GL}{\mathrm{GL}}
\newcommand{\Oo}{\mathrm{O}}
\newcommand{\Pp}{\mathrm{P}}
\newcommand{\PO}{\mathrm{PO}}
\newcommand{\gl}{\mathfrak{gl}}

\newcommand{\Ad}{\mathrm{Ad}}

\newcommand{\SL}{\mathrm{SL}}
\newcommand{\SO}{\mathrm{SO}}
\newcommand{\slfrak}{\mathfrak{sl}}

\newcommand{\Uu}{\mathrm{U}}

\newcommand{\diag}{\mathrm{diag}}

\newcommand{\lfrak}{\mathfrak{l}}

\newcommand{\nfrak}{\mathfrak{n}}
\newcommand{\zfrak}{\mathfrak{z}}
\newcommand{\pfrak}{\mathfrak{p}}
\newcommand{\hfrak}{\mathfrak{h}}
\newcommand{\sfrak}{\mathfrak{s}}
\newcommand{\tfrak}{\mathfrak{t}}

\newcommand{\ofrak}{\mathfrak{o}}

\newcommand{\pgl}{\mathfrak{pgl}}
\newcommand{\pofrak}{\mathfrak{po}}

\newcommand{\dS}{\mathrm{dS}}

\newcommand{\CP}{\mathbb{CP}}
\newcommand{\RP}{\mathbb{RP}}
\newcommand{\PGL}{\mathrm{PGL}}
\newcommand{\Xx}{\mathsf{X}}

\newcommand{\Gg}{\mathsf{G}}

\newcommand{\Hcal}{\mathcal{H}}
\newcommand{\Dcal}{\mathcal{D}}
\newcommand{\Ocal}{\mathcal{O}}

\newcommand{\Ucal}{\mathcal{U}}

\newcommand{\Afrak}{\mathfrak{A}}

\newcommand{\Extalt}{\mathchoice{{\textstyle\bigwedge}}%
    {{\bigwedge}}%
    {{\textstyle\wedge}}%
    {{\scriptstyle\wedge}}}
\newcommand{\sgn}{\mathrm{sgn}}
\newcommand{\glfrak}{\mathfrak{gl}}

\newcommand{\thisd}{{\mathrm{\normalfont d}}}

\newcommand{\totimes}{\,\tilde{\otimes}\,}
\newcommand{\twedge}{\,\tilde{\wedge}\,}
\newcommand{\todot}{\,\tilde{\odot}\,}

\newcommand{\vect}{\mathrm{Vec}}

\newcommand{\vs}{\vspace{0.2cm}}

\tikzset{
    labl/.style={anchor=north, rotate=90, inner sep=1.mm}
}

\usepackage{geometry}
\begin{document}

\title{Quantum fields on projective geometries}
\author[1,2]{Daniel Spitz\footnote{\href{mailto:daniel.spitz@mis.mpg.de}{daniel.spitz@mis.mpg.de}}}
\affil[1]{Institut f\"ur theoretische Physik, Ruprecht-Karls-Universit\"at Heidelberg, \mbox{Philosophenweg 16}, 69120 Heidelberg, Germany}
\affil[2]{Max-Planck-Institut f\"ur Mathematik in den Naturwissenschaften, \mbox{Inselstraße 22, 04103 Leipzig}, Germany}

\date{\today}
\setcounter{Maxaffil}{0}
\renewcommand\Affilfont{\itshape\normalsize}

\maketitle 

\begin{abstract}
Considering homogeneous four-dimensional space-time geometries within real projective geometry provides a mathematically well-defined framework to discuss their deformations and limits without the appearance of coordinate singularities. 
On Lie algebra level the related conjugacy limits act isomorphically to concatenations of contractions.
We axiomatically introduce projective quantum fields on homogeneous space-time geometries, based on correspondingly generalized unitary transformation behavior and projectivization of the field operators.
Projective correlators and their expectation values remain well-defined in all geometry limits, which includes their ultraviolet and infrared limits.
They can degenerate with support on space-time boundaries and other lower-dimensional space-time subspaces.
We explore fermionic and bosonic superselection sectors as well as the irreducibility of projective quantum fields.
Dirac fermions appear, which obey spin-statistics as composite quantum fields.
The framework systematically formalizes and generalizes the ambient space techniques regularly employed in conformal field theory.
\end{abstract}


\section{Introduction}
Space-time geometries set the stage for physical models and can themselves take part in their dynamics. 
We consider four-dimensional homogeneous (Klein) geometries $(\Xx,\Oo)$, which consist of a space-time manifold $\Xx$ and a structure group $\Oo$ acting transitively on $\Xx$.
Physically relevant examples include Poincaré, de Sitter and anti-de Sitter, non-relativistic Galilei and ultra-relativistic Carroll geometries~\cite{Bergshoeff:2022eog}.

Limits of deformed geometries can naturally appear for general-relativistic space-times, e.g.~degenerate near-horizon limits around black holes, or steadily diluted matter, for which curved geometries transform towards flat Poincaré geometry.
In physics, such limits are often taken by contracting the Lie algebras of structure groups, i.e., by consistently sending individual commutators to zero~\cite{inonu1953contraction,Bergshoeff:2022eog}. 
Given the setting of geometries $(\Xx,\Oo)$, commutators of contracted Lie algebras of structure groups are in general not matrix commutators anymore and space-time manifolds remain inconsistently unaltered.
Instead, many four-dimensional geometries $(\Xx,\Oo)$ can be considered within the real projective geometry $(\RP^4,\PGL_5\rr)$. 
Then a canonical framework exists to discuss their deformations and limits~\cite{cooper2018limits}, based on matrix products. 
In the projective setting, geometries remain well-defined in limits and merely degenerate.
Beyond the mentioned examples, this allows for the description of finite and infinite dilatations in specific space-time directions~\cite{Gutowski:2004ez, Kaul:2016lhx, Bagchi:2022nvj, Faedo:2022hle} and scale transformations as for renormalization group studies~\cite{salmhofer2007renormalization}. 
We show that for Lie algebras of many structure groups such limits are isomorphic to compositions of contractions.
Considering projective geometric structures, no infinite blow-ups occur for instance for projective tensor fields.

Holographic correspondences provide intriguing instances, where information of fields in the bulk of a higher-dimensional space-time can be represented in terms of (non-local) degrees of freedom on its lower-dimensional boundary.
Most prominently, in the anti-de Sitter/conformal field theory (AdS/CFT) correspondence~\cite{Maldacena:1997re} supergravity in the bulk of $\mathrm{AdS}^5\times S^5$ is linked to supersymmetric Yang-Mills theory on its boundary.
A number of works have explored flat limits of the AdS/CFT correspondence, see e.g.~\cite{Susskind:1998vk,Fitzpatrick:2010zm,Hijano:2020szl}.
Against this background, given a somewhat general quantum field theory (QFT), it would be beneficial to enhance the understanding of the form of correlation functions in geometry limits within a unified, mathematically concise framework.

This motivates the introduction of projective quantum fields on the ambient projective geometry $(\RP^4,\PGL_5\rr)$ in this work, which encompasses aspects of the more traditional cases of quantum fields on fixed geometries such as Poincaré or de Sitter geometry via geometry restriction.
For their description, we note that QFTs can be suitably formulated on curved space-time geometries in non-rigorous~\cite{Parker:2009uva} and in algebraic approaches~\cite{Wald:1995yp, Hollands:2014eia, Fredenhagen:2014lda}.
We focus on the involved representation-theoretic structures in an axiomatic approach by means of generalizing part of the Wightman axioms~\cite{Haag1996Local}.
To enhance the comparability with common, non-rigorous QFT formulations, the formulation of projective quantum fields is based on the existence of a Hilbert space, even if for a given set of field operators unitarily inequivalent Hilbert space formulations can exist~\cite{Wald:1995yp, Haag1996Local}.
These can all be considered equivalent for the derivations in this work, since at no point canonical commutation relations and an explicit form of the Hilbert space are employed.
Furthermore, we do not consider implications of global hyperbolicity of the space-time manifold, causality and the spectrum condition, which are therefore not assumed.
\vs

\begin{defi*}[Shortened \Cref{DefFieldOperator}]
A \emph{projective quantum field} $(U, \rho, \{[\hat{\Ocal}([x])]\allowbreak |\,[x]\in \RP^4\})$ consists of a projective unitary $\PGL_5\rr$ representation $U$ on a Hilbert space $\Hcal$ with adjoint~${}^\dagger$, a finite-dimensional complex $\overline{\PGL_5\rr}$ representation $\rho$, and an equivalence class of tuples of operator-valued tempered distributions $[\hat{\Ocal}_1([x]),\dots,\allowbreak \hat{\Ocal}_{\dim\rho}([x])]$ defined modulo $C^\infty(\RP^4,\rr_{\neq 0})$ prefactors, with domain a dense subspace of $\Hcal$ and their smeared representatives satisfying a boundedness criterion as specified in the main text, such that for all $[x]\in\RP^4$, $[g]\in \PGL_5\rr$, $\alpha=1,\dots,\dim\rho$:
\begin{equation*}
U([g])\hat{\Ocal}_\alpha([x])U^\dagger([g]) = \sum_{\beta=1}^{\dim\rho}\rho_{\alpha\beta}([g^{-1}]) \, \hat{\Ocal}_\beta([g\cdot x])\,.
\end{equation*}
On the geometry $(\Xx,\Oo) <  (\RP^4,\PGL_5\rr)$, the projective quantum field is given by restriction:
\begin{equation*}
\hat{\Ocal}|_{(\Xx,\Oo)}:= (U|_\Oo,\rho,\{[\hat{\Ocal}([x])]\,|\, [x]\in\Xx\})\,.
\end{equation*}
\end{defi*}
\vs

We show that such projective quantum fields transform naturally and smoothly under deformations and limits of geometries, but their support can shift to lower-dimensional space-time closure subspaces.

Given a projective quantum field, its projective correlators can be considered, i.e., projectivized tensor products of smeared field operator representatives, which form by construction bounded operators.
Concerning their behavior under geometry deformations and limits, we show the following.
\vs

\begin{thm*}[Shortened \Cref{ThmCorrelatorDegenerationInLimits}]
Projective correlators remain bounded in limits of geometries.
If the space-time geometry degenerates in the limit process, the limiting projective correlators are degenerate as finite-rank tensors with field operator components and have support only on space-time boundaries and other lower-dimensional space-time subspaces.
\end{thm*}
\vs

Therefore, corresponding QFTs can reduce dimensionally in geometry limits. 
We describe geometry limits of projective correlators for a range of physically relevant examples.
This provides insights into the possible imprints of projective correlators after such limit processes, e.g.~the flat Poincaré limit of projective correlators on de Sitter and anti-de Sitter geometries.
We prove that projective correlation functions, i.e., expectation values of projective correlators, are well-defined for all states in $\Hcal$, also in geometry limits.
On a formal level, this implies that all projective correlation functions have finite infrared and ultraviolet limits, for which the corresponding QFT reduces to three space-time dimensions.

The usage of ambient space techniques for quantum fields and their correlation functions actually dates back to Dirac~\cite{Dirac:1935zz, Dirac:1936fq}, which can, at least in parts, be seen as a special case of our systematic framework.
Such methods have been employed in recent decades in AdS/CFT studies, where four-dimensional fields are constructed as projections of fields on a hypercone in six-dimensional projective space, see e.g.~\cite{Boulware:1970ty, Weinberg:2010fx}.
Amongst others, this can facilitate the computation of operator product expansions~\cite{Ferrara:1973yt, Ferrara:1973eg}, higher-spin conformal correlators~\cite{Costa:2011mg} and boundary values of AdS gauge fields~\cite{Bekaert:2012vt, Bekaert:2013zya}.
\vs

Concerning the classification of projective quantum fields, fermionic and bosonic superselection sectors can be extended to projective quantum fields on the ambient geometry $(\RP^4,\PGL_5\rr)$.
We construct composite projective quantum fields and characterize them according to irreducibility of the Lie algebra representation $\tilde{\rho}$ corresponding to $\rho$.
All irreducible representations of $\pgl_5\rr\cong \slfrak_5\rr$ are given by Schur modules, which leads to the conclusive classification of projective quantum fields.
This is analogous to the classification of quantum fields on Poincaré geometry according to spin.

Projective quantum fields can as well be characterized according to irreducibility under Poincaré transformations.
This leads to the following result.
\vs

\begin{thm*}[\Cref{ThmSpinStatIrredPoincareIrred}]
Fermionic, irreducible, Poincaré-irreducible projective quantum fields behave under Poincaré transformations as Dirac fermions and obey spin-statistics as composite projective quantum fields.
If they are bosonic, they behave under Poincaré transformations as scalar or vector bosons, but violate spin-statistics as composite projective quantum fields.
\end{thm*}
\vs

In the given framework, this singles out Dirac fermions.
Yet, projective quantum fields may as well give rise to higher-spin Poincaré group representations upon restriction to Poincaré geometry, if $\tilde{\rho}$ is not irreducible.
\vs

This manuscript is organized as follows.
In \Cref{SecGeometriesLimits} we introduce (homogeneous) geometries and their deformations and limits, discuss the relation between conjugacy limits and contractions, and consider projective tensor fields. 
\Cref{SecQFTOnDeformedGeometries} is devoted to projective quantum fields.
Projective correlator algebras, fermionic and bosonic superselection sectors, composite projective quantum fields as well as their $\pgl_5\rr$- and Poincaré-irreducibility are explored.
\Cref{SecConclusions} provides further questions.
Some of the proofs are deferred to a number of appendices, which are given in the end.


\section{Geometries and limits}\label{SecGeometriesLimits}
Mathematically, a geometry in the sense of Klein is defined as a pair $(\Xx,\Oo)$ of a smooth, connected manifold $\Xx$ (the model space) and a Lie group $\Oo$ (the structure group) acting transitively by analytic maps on it, i.e., for any pair of points ${x,y\in \Xx}$ there exists $g\in \Oo$, such that $y=g\cdot x$. 
Transitivity identifies the model space $\Xx$ with the homogeneous space $\Oo/\Oo_x$, where $\Oo_x$ denotes the stabilizer of an arbitrary point $x\in \Xx$.
For more details on the category of such geometries we refer to~\cite{cooper2018limits}.

Similarly, the geometry of a homogeneous physical space-time includes both a space-time manifold and a Lie group of space-time symmetry transformations. 
Space-time symmetry groups often contain the indefinite orthogonal groups $\Oo(p,q)$, i.e., the groups of isometries of the metric tensor
\begin{equation}\label{EqMetric}
-(\dd x^1)^2 - \ldots - (\dd x^{p})^2 + (\dd x^{p+1})^2 + \ldots + (\dd x^{p+q})^2\,,
\end{equation}
with examples the Lorentz group $\Oo(3,1)$ or the de Sitter group $\Oo(4,1)$. 
Here, $\dd x^1,\dots,\allowbreak \dd x^{p+q}$ denote the canonical 1-forms on $\rr^{p+q}$, and $(\dd x^j)^2 := \dd x^j \odot \dd x^j$ with $\odot$ the symmetrized tensor product.
More generally, groups such as
\begin{equation*}
\Oo((p_0,q_0),\dots,(p_k,q_k)):=\left(\begin{matrix}
\Oo(p_0,q_0) & 0 & \cdots & 0 \\
\rr^{(p_1+q_1)\times (p_0+q_0)} & \Oo(p_1,q_1) &  & 0\\
\vdots & \vdots & \ddots & \vdots \\
\rr^{(p_k+q_k)\times (p_0+q_0)} & \rr^{(p_k + q_k)\times (p_1+q_1)} & \cdots & \Oo(p_k,q_k)
\end{matrix}\right)
\end{equation*}
can appear. 
To prevent singular behavior as discussed later, we consider projective geometries with projective structure groups $\PO((p_0,q_0),\dots,(p_k,q_k))$, defined modulo multiplication by non-zero constants. 
For introductions to projective geometry we refer to~\cite{ovsienko2004projective,Gallier2011}.
We note that $\PO((1),(3,1))$ is isomorphic to the Poincaré group in 3+1 space-time dimensions, the latter isomorphically embedded in $\Pp\GL_5\rr$ as follows.
To see this, let $\Lambda\in\Oo(3,1)$, $t\in \rr^4$.
Then the projective $5\times 5$ matrices
\begin{equation}\label{EqPoincareTrafo55}
[(\Lambda, t)]:=\Pp\left(\begin{matrix}
1 & \\
t & \Lambda
\end{matrix}\right)\in \PO((1),(3,1))
\end{equation}
together with (projective) matrix multiplication furnish a representation of the Poincaré group, which is isomorphic to the latter.
The elements $t\in \rr^4$, which appear in the lower-left corner of $[(\Lambda,t)]$, correspond to the usual space-time translations.

The group $\PO((p_0,q_0),\dots,(p_k,q_k))$ with $p_0+q_0+\ldots + p_k+q_k=m$ acts transitively on the model space
\begin{align*}
&\Xx((p_0,q_0),\dots,(p_k,q_k))\nonumber\\
& := \{[x_0,\dots,x_{m-1}]\in \RP^{m-1}\,|\, -x_0^2-\ldots-x_{p_0-1}^2+x_{p_0}^2+\ldots+x_{p_0+q_0-1}^2 < 0\}\,.
\end{align*}
Therefore, the pairs 
\begin{equation*}
\Gg((p_0,q_0),\dots,(p_k,q_k)):= (\Xx((p_0,q_0),\dots,(p_k,q_k)),\PO((p_0,q_0),\dots,(p_k,q_k)))
\end{equation*}
form geometries.
We mostly restrict to four-dimensional space-time geometries in this work ($m=5$).
For instance, $\Gg(4,1)$ with model space $\Xx(4,1)$ is the projective model of four-dimensional de Sitter geometry and $\Gg(3,2)$ with model space $\Xx(3,2)$ is the projective model of four-dimensional anti-de Sitter geometry, while $\Gg((1),(3,1))$ with
\begin{equation}\label{EqPoincareModelSpaceA31}
\Xx((1),(3,1)) = \{[x_0,\dots,x_4]\,|\, x_0\neq 0\} = \Abb^{3,1}
\end{equation}
is the projective model of Poincaré geometry, which can be identified with Poincaré geometry itself. 
Here, $\Abb^{p,q}$ denotes the $(p+q)$-dimensional affine space with the $(p,q)$-signature metric tensor \eqref{EqMetric}, and the identification with $\Abb^{3,1}$ in \Cref{EqPoincareModelSpaceA31} is via 
\begin{equation*}
[x_0,\dots,x_4]= [1,x_1/x_0,\dots,x_4/x_0]\mapsto (x_1/x_0,\dots,x_4/x_0)\,.
\end{equation*}
Poincaré geometry also comes with the metric tensor $-(\dd y^1)^2 - (\dd y^2)^2 - (\dd y^3)^2 + (\dd y^4)^2$, $y_\mu = x_\mu/x_0$ for $\mu=1,\dots,4$, which is the usual Minkowski metric tensor. 

The group $\PO((1),(1),(3))$ can be identified with the group of Galilei transformations, and $\PO((1),(3),(1))$ is isomorphic to the group of Carroll transformations.
Indeed, the matrix multiplications are (ignoring $\pm 1$ entries on the diagonals)
\begin{align*}
\Pp\left(\begin{matrix}
1 &&\\
s & 1 &\\
\mathbf{a} & \mathbf{v} & R
\end{matrix}\right)\cdot \Pp\left(\begin{matrix}
1 &&\\
s' & 1 &\\
\mathbf{a}' & \mathbf{v}' & R'
\end{matrix}\right) = &\;\, \Pp\left(\begin{matrix}
1 &&\\
s+s' & 1 &\\
\mathbf{a} + s'\mathbf{v} + R \mathbf{a}' & \mathbf{v} + R\mathbf{v}' & RR'
\end{matrix}\right)\,,\\
\Pp\left(\begin{matrix}
1 &&\\
\mathbf{a} & R & \\
s & \mathbf{b}^T & 1
\end{matrix}\right)\cdot \Pp\left(\begin{matrix}
1 &&\\
\mathbf{a}' & R' & \\
s' & \mathbf{b}'^T & 1
\end{matrix}\right)= &\;\, \Pp\left(\begin{matrix}
1 &&\\
\mathbf{a}+R\mathbf{a}' & RR' &\\
s+s'+\mathbf{b}^T\mathbf{a}' & \mathbf{b}^T R'+\mathbf{b}'^T & 1
\end{matrix}\right)\,,
\end{align*}
where $s,s'\in \rr$, $\mathbf{a},\mathbf{a}',\mathbf{b},\mathbf{b}',\mathbf{v},\mathbf{v}'\in\rr^3$, $R,R'\in \Oo(3)$ and which are matrix versions of the group multiplications of the Galilei and Carroll groups, respectively~\cite{Duval:2014uoa}.
Galilei geometry $\Gg((1),(1),(3))$ has the model space $\Xx((1),(1),(3)) = \Abb^1\times \Abb^{3}$ and comes with degenerate Galilean structures, while ultra-relativistic Carroll geometry $\Gg((1),(3),(1))$ has the model space $\Xx((1),(3),(1))= \Abb^3\times \Abb^1$ and comes with degenerate Carroll structures~\cite{Ciambelli:2019lap, Bergshoeff:2022eog}.%
\footnote{Actually, all geometries $\Gg((p_0,q_0),\dots,(p_k,q_k))$ come with certain affine bundle structures as detailed in~\cite{cooper2018limits}, thereby generalizing structures such as Galilean and Carroll structures.}

\subsection{Deformations and limits}
Geometry limits such as the Poincaré limit of de Sitter space-time, which incorporates the flattening of $\dS^4$ to $\Abb^{3,1}$, or the non-relativistic (Galilei) limit of Poincaré geometry can be described in terms of projective geometries. 
More specifically, consider four-dimensional geometries $\Gg=(\Xx,\Oo)$, which are subgeometries of the ambient geometry $(\RP^4,\PGL_5\rr)$, i.e., $\Xx$ is an open submanifold of $\RP^4$ and $\Oo$ is a subgroup of $\PGL_5\rr$ acting transitively on $\Xx$.
$\Gg(4,1)$, $\Gg((1),(3,1))$, $\Gg((1),(1),(3))$ and $\Gg((1),(3),(1))$ are all of this type, for instance.
Let $[b_n]\in \Pp\GL_5\rr$ be a sequence of group elements. 
It acts on points $[x]\in \Xx$ via (projective) matrix multiplication: $[x]\mapsto [b_n \cdot x]$. 
Group elements $[h]\in \Oo$ are conjugated by the $[b_n]$: $[h]\mapsto \Ad_{[b_n]}[h] = [b_n  h  b_n^{-1}]$, analogously to a change of basis acting on matrices. 
The $[b_n]$ thus act on a geometry $\Gg$ as
\begin{equation*}
[b_n]_*:\Gg=(\Xx,\Oo)\to [b_n]_*\Gg:=([b_n]\cdot \Xx,\Ad_{[b_n]}\Oo)
\end{equation*}
with $\Ad_{[b_n]}\Oo:=[b_n]\cdot \Oo \cdot [b_n^{-1}]$.
We call $[b_n]_*\Gg$ a geometry deformation of $\Gg$.

Following~\cite{cooper2018limits}, the sequence $\Ad_{[b_n]}\Oo$ of Lie groups converges geometrically to another Lie subgroup $\Oo'<\PGL_5\rr$ for $n\to\infty$, if every $[h']\in \Oo'$ is the limit of some sequence $[h_n]\in \Ad_{[b_n]}\Oo$ and if every accumulation point of every sequence $[h_n]\in \Ad_{[b_n]}\Oo$ lies in $\Oo'$.
Then $\Oo'$ is called the conjugacy limit of $\Oo$ via $[b_n]$. 
Conjugacy limits are limits on the space of closed subgroups of $\PGL_5\rr$ with respect to the so-called Chabauty topology; for details see e.g.~\cite{fisher2009space1,fisher2009space2,cooper2018limits,biringer2018metrizing}. 
The sequence $[b_n]_*\Gg\subset (\RP^4,\PGL_5\rr)$ converges to the geometry $\Gg'=(\Xx',\Oo')<(\RP^4,\PGL_5\rr)$, if $\Ad_{[b_n]}\Oo$ converges geometrically to $\Oo'$ and there exists%
\footnote{This condition exemplifies the \emph{hit-and-miss} character of the Chabauty topology~\cite{trettel2019families}.}
 $[z]\in \Xx'$, such that for all $n$ sufficiently large: $[z]\in [b_n]\cdot \Xx$. 
We also call the latter topology on the space of model spaces the model space topology.
All such limits of subgeometries $\Gg((p_0,q_0),\dots,(p_k,q_k))< (\RP^4,\PGL_5\rr)$ have been classified and are up to geometry deformations in $(\RP^4,\PGL_5\rr)$ of the form $\Gg((p_0',q_0'),\dots,(p_l',q_l'))$ for $l\geq k$, $p_0+\ldots+p_k = p_0'+\ldots+p_l'$ and $q_0+\ldots+q_k = q_0'+\ldots+q_l'$, potentially after exchanging some $(p_i,q_i)$ with $(q_i,p_i)$~\cite{cooper2018limits}. 
Such a limit geometry is also called a refinement or degeneration of $\Gg((p_0,q_0),\dots,(p_k,q_k))$.

Elements of the Lie algebra $\pgl_5\rr$ and thus also of the Lie algebra $\ofrak$ of $\Oo<\PGL_5\rr$ are defined modulo the addition of real multiples of the $5\times 5$ identity matrix.%
\footnote{We employ mathematical conventions for Lie algebras, such that a connected Lie group $G$ has Lie algebra $\mathfrak{g}$ and $G$ is generated by $\exp(\mathfrak{g})$.}
If a Lie algebra element $[X]\in\pgl_5\rr$ acts on a $\PGL_5\rr$ group element or vice versa, its unique trace-zero representative $X\in\slfrak_5\rr\cong \pgl_5\rr$, taken modulo non-zero real prefactors, is considered.
Then the sequence $[b_n]$ acts on Lie algebra elements analogously to their action on group elements: $[X]\mapsto \Ad_{[b_n]}[X]$ for $[X] \in \ofrak$. 
With $\Oo\to \Oo'$ the conjugacy limit in $\PGL_5\rr$ via $[b_n]$ and $\ofrak'=\mathrm{Lie}(\Oo')$, it has been shown in~\cite{trettel2019families} that
\begin{equation*}
\ofrak' = \lim_{n\to\infty} [b_n]\cdot \ofrak\cdot [b_n^{-1}]\,,
\end{equation*}
convergence defined as for the geometric convergence of Lie groups. 
\vs

\begin{example}\label{ExampleGeometryDeformations}
\begin{enumerate}[(i)]
\item To explore the deformation of projective Poincaré geometry $\Gg((1),\allowbreak (3,1))$ by
\begin{equation}\label{EqNonRelLimitMatrices}
[b_n] = \Pp\left(\begin{matrix}
e^{-n} & &\\
 & 1_{3\times 3} &\\
& & e^{-n}
\end{matrix}\right)\,,
\end{equation}
$1_{3\times 3}$ the $3\times 3$ unit matrix, focus on a boost generator of the Poincaré Lie algebra such as
\begin{equation}\label{EqLorentzBoost}
[K] = \Pp\left(\begin{matrix}
0 &&&&\\
 & 0 &   &   &  \\
 &   & 0 &   &  \\
 &   &   & 0 & 1\\
 &   &   & 1 & 0
\end{matrix}\right)\,.
\end{equation}
The $[b_n]$ act via conjugation on $[K]$, which yields for the non-trivial submatrix:
\begin{equation}\label{EqConjugationPoincareGalileiRewriting}
\Pp\left(\begin{matrix}
0 & 1\\
1 & 0 
\end{matrix}\right) \mapsto \Pp\left(\begin{matrix}
0 & e^{n}\\
e^{-n}& 0
\end{matrix}\right) = \Pp\left(\begin{matrix}
0 & 1\\
e^{-2n} & 0
\end{matrix}\right)\,.
\end{equation}
This has the well-defined $n\to\infty$ limit
\begin{equation}\label{EqGalileiBoost}
\Pp\left(\begin{matrix}
0 & 1 \\
0 & 0 
\end{matrix}\right)\,.
\end{equation}
The multiplication by $\exp(-n)$ employed in \Cref{EqConjugationPoincareGalileiRewriting} is an identity map in the projective setting. 
In a non-projective setting, the limit matrix would contain diverging elements.
Analogously taking the conjugacy limits of the other generators of $\pfrak\ofrak((1),(3,1))$ shows that the structure groups $\Ad_{[b_n]}\PO((1),(3,1))$ converge geometrically for $n\to\infty$ to
\begin{equation*}
\lim_{n\to\infty}\Ad_{[b_n]}\PO((1),(3,1)) = \Ad_{\tau}\PO((1),(1),(3))\,,
\end{equation*}
where $\tau=(0)\,(1\, 2\, 3\, 4)$ (in cycle notation) is a coordinate index permutation.
This is isomorphic to the Galilei group.
The limit process turns Lorentz boosts with generators such as \eqref{EqLorentzBoost} into Galilean velocity additions as generated by the projective $5\times 5$ matrix corresponding to \eqref{EqGalileiBoost}. 
The model space $\Xx((1),(3,1))$ is invariant under the action of the $[b_n]$ of \Cref{EqNonRelLimitMatrices}, such that for $n\to\infty$ projective Galilei geometry is retrieved up to a coordinate permutation:
\begin{equation*}
[b_n]_* \Gg((1),(3,1)) \to \tau_*\Gg((1),(1),(3)):=(\tau\cdot \Xx((1),(3,1)), \Ad_\tau \PO((1),(1),(3)))\,.
\end{equation*}
Note that permutations act on model spaces and structure groups inversely to how they permute the index arguments of geometries.
\item Consider
\begin{equation}\label{EqDeSitterToPoincarecn}
[b_n] = \Pp\left(\begin{matrix}
e^{-4n} &\\
& e^n\cdot 1_{4\times 4}
\end{matrix}\right)
\end{equation}
acting on projective de Sitter geometry $\Gg(4,1)$. 
Let $[x_0,\dots,x_4]\in [b_n]\cdot\Xx(4,1)$, i.e.,
\begin{equation*}
-e^{8n} x_0^2 - e^{-2n} (x_1^2 +x_2^2 + x_3^2 - x_4^2) < 0\,.
\end{equation*}
This is equivalent to
\begin{equation*}
- x_0^2 - e^{-10n}(x_1^2  + x_2^2 + x_3^2 -x_4^2) < 0\,,
\end{equation*}
yielding the well-defined limit constraint $x_0 \neq 0$ for $n\to\infty$. 
Thus, the deformed model space $[b_n]\cdot \Xx(4,1)$ changes in the $n\to\infty$ limit to $\Xx((1),(3,1))= \Abb^{3,1}$.
Explicit computation shows that the deformed structure groups $\Ad_{[b_n]} \PO(4,1)$ converge geometrically to $\PO((1),(3,1))$ for $n\to\infty$. 
Therefore, the limit process flattens the de Sitter model space and the structure group changes accordingly. 
The limit geometry of projective de Sitter geometry, deformed by the $[b_n]$ of \Cref{EqDeSitterToPoincarecn}, is projective Poincaré geometry $\Gg((1),(3,1))$. 
\item Similarly to (ii), consider
\begin{equation}\label{EqAntiDeSitterToPoincarecn}
[b_n] = \Pp\left(\begin{matrix}
e^n\cdot 1_{4\times 4} & \\
& e^{-4n}
\end{matrix}\right)
\end{equation}
acting on projective anti-de Sitter geometry $\Gg(3,2)$.
A point $[x_0,\ldots,x_4]$ is in $[b_n]\cdot \Xx(3,2)$, if and only if 
\begin{equation*}
e^{-10n} (-x_0^2-x_1^2-x_2^2+x_3^2) + x_4^2 < 0\,,
\end{equation*}
which yields the constraint $x_4\neq 0$ for $n\to \infty$.
The sequence $\Ad_{[b_n]}\PO(3,2)$ converges with respect to Chaubauty topology to
\begin{equation*}
\Pp\left(\begin{matrix}
\Oo(3,1) & \rr^4\\
& \pm 1
\end{matrix}\right)
\end{equation*}
as $n\to\infty$.
Therefore, up to the permutation $\sigma = (0\, 1\, 2\, 3\, 4)$ this yields projective Poincaré geometry for $n\to\infty$:
\begin{equation*}
[b_n]_*\Gg(3,2)\to \sigma_* \Gg((1),(3,1))\,.
\end{equation*}
\end{enumerate}
\end{example}

\begin{remark}
While the choice of $(\RP^4,\PGL_5\rr)$ as the ambient geometry is not unique, the results of this work are to some extent independent from the choice of ambient geometries of type $(\RP^{m-1},\PGL_m\rr)$ for $m\geq 5$, see \Cref{AppendixAmbientGeometryChoice}. 
Setting $m=5$ can be seen as a minimal choice.
A few works have considered conjugacy limits of subgroups within other groups than $\PGL_m\rr$ or $\GL_m\rr$, see e.g.~\cite{leitner2016conjugacy, leitner2016limits, lazarovich2021local}.
\end{remark}

\subsection{Conjugacy limits and contractions}\label{AppendixLimitsLieAlgebras}
In physics, contractions of the Lie algebras of structure groups are often considered instead of conjugacy limits.
Let $\ofrak'$ be a conjugacy limit of $\ofrak$ on Lie algebra level.
First, we give examples of physical interest, for which $\ofrak'$ is isomorphic to a contraction of~$\ofrak$.
We also provide an example of a contraction, which cannot be isomorphically described as a conjugacy limit.
Subsequently, we show that there are conjugacy limits, which are not isomorphic to single contractions on Lie algebra level.
A theorem demonstrates that the composition of contractions can always achieve this.
In this subsection, we sometimes consider geometries of general dimension.

We define a contraction as in~\cite{dooley1985contractions}. 
For this, consider a Lie subalgebra $\tfrak$ of a Lie algebra $\hfrak$, the latter with underlying vector space $\vect(\hfrak)$. 
Let $\tfrak^c$ denote the subspace of $\vect(\hfrak)$ complementary to $\tfrak$, such that $X\in \vect(\hfrak)$ can be uniquely written as $X=X_\tfrak + X_{\tfrak^c}$, $X_\tfrak \in \tfrak$, $X_{\tfrak^c}\in \tfrak^c$.
For $\varepsilon>0$ define $\phi_\varepsilon(X):= X_\tfrak + \varepsilon X_{\tfrak^c}$.
Then $[X,Y]':=\lim_{\varepsilon\to 0}\phi_\varepsilon^{-1}([\phi_{\varepsilon}(X),\phi_\varepsilon(Y)])$ exists for all $X,Y\in \vect(\hfrak)$.
The vector space $\vect(\hfrak)$ equipped with the commutator $[\cdot,\cdot]'$ is called the contraction of $\hfrak$ along $\tfrak$ and forms a Lie algebra.
We note that there exist other definitions in the literature, see e.g.~\cite{nesterenko2006contractions}.
\vs

\begin{example}\label{ExampleConjLimitsContractionsIso}
The Lie algebra $\ofrak(3)$ has generators
\begin{equation*}
X^1 = \left(\begin{matrix}
0 & 1 & 0\\
-1 & 0 & 0\\
0 & 0 & 0
\end{matrix}\right), \quad X^2 = \left(\begin{matrix}
0 & 0 & -1\\
0& 0 & 0\\
1 & 0 & 0
\end{matrix}\right),\quad X^3 = \left(\begin{matrix}
0 & 0 & 0\\
0& 0 & 1\\
0 & -1 & 0
\end{matrix}\right)
\end{equation*}
with commutators $[X^1,X^2] = X^3$, $[X^1,X^3]  =- X^2$, $[X^2,X^3] = X^1$. 
We consider its projective variant $\pofrak(3)$ and the conjugacy limit via
\begin{equation*}
[b_n] = \Pp\left(\begin{matrix}
e^{-2n} &&\\
 & e^n &\\
 && e^n
\end{matrix}\right)\,,
\end{equation*}
which yields the projective Euclidean motion Lie algebra in two dimensions, $\pofrak((1),(2))$.
The non-projective $\ofrak((1),(2))$ has generators
\begin{equation}\label{EqY1Y2ExampleContraction}
Y^1 = \left(\begin{matrix}
0 & 0 & 0\\
-1 & 0 & 0\\
0 & 0 & 0
\end{matrix}\right),\quad  Y^2 = \left(\begin{matrix}
0 & 0 & 0\\
0& 0 & 0\\
1 & 0 & 0
\end{matrix}\right)
\end{equation}
and $X^3$ along with matrix commutators. 
We contract $\ofrak(3)$ along $X^3$, which yields the contracted commutators $[X^1,X^2]'=0$, $[X^1,X^3]' = - X^2$ and $[X^2,X^3]' = X^1$. 
After descending to its projective variant, the contracted Lie algebra is isomorphic to the Lie algebra $\pofrak((1),(2))$ of the conjugacy limit.

We take a final conjugacy limit of $\pofrak((1),(2))$ via 
\begin{equation*}
[b_n] = \Pp\left(\begin{matrix}
e^{-n} &&\\
 & e^{-n} &\\
 && e^{2n}
\end{matrix}\right)\,,
\end{equation*}
which yields the projective Heisenberg Lie algebra $\pofrak((1),(1),(1))$. 
Its non-projective variant is generated by $Y^1,Y^2$ and
\begin{equation}\label{EqY3ExampleContraction}
Y^3  = \left(\begin{matrix}
0 & 0 & 0\\
0 & 0 & 0\\
0 & -1 & 0
\end{matrix}\right)\,,
\end{equation}
such that $[Y^1,Y^2] = 0$, $[Y^1,Y^3] = - Y^2$ and $[Y^2,Y^3] = 0$. 
We obtain a Lie algebra isomorphic to $\pofrak((1),(1),(1))$ via contraction of $\ofrak((1),(2))$ along $X^1$ and projectivization.
\end{example}
\vs

Continuing \Cref{ExampleConjLimitsContractionsIso} shows that not every contraction corresponds to a conjugacy limit.
\vs

\begin{example}\label{ExampleContractionNotIsoConj}
We contract the Heisenberg Lie algebra $\ofrak((1),(1),(1))$ with generators $Y_1,Y_2,Y_3$ as in \Cref{EqY1Y2ExampleContraction,EqY3ExampleContraction} along $\mathrm{Span}(Y_2)$, which is the center of $\ofrak((1),(1),(1))$. 
The contracted commutators are all trivial, such that the contracted Lie algebra is isomorphic to $\rr^3$.
The Lie algebra $\pofrak((1),(1),(1))$ is a final conjugacy limit of $\pofrak(3)$ up to conjugacy in $\PGL_3\rr$.
No further non-trivial refinement is possible~\cite{cooper2018limits}. 
Therefore, the contracted algebra $\rr^3$ does not correspond to a conjugacy limit of $\pofrak(3)$ via sequences in $\PGL_3\rr$.
\end{example}
\vs

Even if examples of isomorphisms between Lie algebras of conjugacy limits and corresponding contractions can be given, such isomorphisms do not generally exist, as the following lemma demonstrates.
\vs

\begin{lem}\label{LemmaNoContractionO3toO111}
For $m\geq 3$ no single contraction of the Lie algebra $\pofrak(m)$ yields a Lie algebra isomorphic to $\pofrak((1),\dots,(1))$ ($m$ copies of $(1)$ as arguments).
\end{lem}
\vs

The Lie algebra $\pofrak((1),\dots,(1))$ is the conjugacy limit of $\pofrak(m)<\pgl_m\rr$ via $[b_n] = [\diag(1,\exp(n),\dots,\exp((m-1)n))]$.
The proof of \Cref{LemmaNoContractionO3toO111} is postponed to \Cref{AppendixPfLemmaNoContractionO3toO111}.

It is clear from \Cref{ExampleConjLimitsContractionsIso} that the concatenation of multiple contractions of $\pofrak(m)$ can yield a Lie algebra isomorphic to $\pofrak((1),\dots,(1))$.
This exemplifies the following theorem.
\vs

\begin{thm}\label{ThmMultipleContrations}
Consider $\PO(p,q)<\PGL_{p+q}\rr$ and let $\Oo'$ be the conjugacy limit of $\PO(p,q)$ via a sequence $[b_n]\in \PGL_{p+q}\rr$ for $n\to\infty$.
Then the Lie algebra $\ofrak'$ of $\Oo'$ is isomorphic to the composition of a finite number of contractions of $\pofrak(p,q)$.
\end{thm}
\vs

The constructive proof of this theorem is deferred to \Cref{AppendixPfThmMultipleContrations}.

\subsection{Projective frames and tensor fields}\label{SecProjTensorFields}
Analogously to conjugacy limits of projective subgroups, projective variants of tensor fields do not exhibit singular behavior in limits of geometries, but merely degenerate as discussed now.
Specifically, we define projective vector fields on a four-dimensional model space $\Xx$ as equivalence classes $[W]$ of vector fields on $\Xx$, which are defined up to the equivalence relation $W\sim \tilde{W}$, if there exists $\lambda\in C^\infty(\Xx,\rr_{\neq 0})$ such that for all $[x]\in\Xx$: $\tilde{W}([x]) = \lambda([x])\, W([x])$, in short $\tilde{W}=\lambda W$.
We note that for connected model spaces $\Xx$, the functions $\lambda\in C^\infty(\Xx,\rr_{\neq 0})$ have either positive or negative values on all $\Xx$.
\vs

\begin{remark}
This definition of projective vector fields is consistent with the more common definition of projectively related, torsion-free connections. 
According to~\cite{whitehead1931representation, tanaka1957projective, Eastwood2008}, two torsion-free connections on $\Xx$ are projectively related, if they have the same geodesics as point sets, i.e., if their geodesics agree up to parametrization.
Consider a vector field $W$ on $\Xx$, let $\lambda\in C^\infty(\Xx,\rr_{\neq 0})$ and set $\tilde{W}=\lambda W$.
Given a torsion-free connection $\nabla$, for every $[x]\in \Xx$ there is a unique geodesic $\gamma:(-\epsilon,\epsilon)\to \Xx$, $\epsilon>0$, with $\gamma(0) = [x]$ and $\dd \gamma(s)/\dd s|_{s=0} = W([x])$.
Similarly, $\tilde{W}$ induces a geodesic $\tilde{\gamma}:(-\tilde{\epsilon},\tilde{\epsilon})\to \Xx$, $\tilde{\epsilon}>0$, with $\tilde{\gamma}(0) = [x]$ and
\begin{equation*}
\frac{ \thisd }{\thisd s}\bigg|_{s=0}\tilde{\gamma}(s) = \tilde{W}([x]) = \lambda([x]) W([x]) = \lambda([x])\frac{ \thisd }{\thisd s}\bigg|_{s=0}\gamma(s)\,,
\end{equation*}
so that the geodesics $\gamma$ and $\tilde{\gamma}$ are the same as sets.
Then there exists a torsion-free connection $\tilde{\nabla}$ with geodesic $\tilde{\gamma}$, such that $\nabla$, $\tilde{\nabla}$ are projectively related~\cite{Eastwood2008}.
Therefore, projectively equivalent vector fields are indeed consistent with the notion of projectively related connections.
\end{remark}
\vs

Locally, projective vector fields can be viewed as projective linear combinations of representatives of projective frame vector fields.
A projective frame on the four-dimensional space $\Xx$ is defined from a frame $(e_1,\dots, e_4)$ on $\Xx$ as 
\begin{equation*}
\{([x],\{[e_1([x])],\dots, [e_4([x])],[e_1([x])+\ldots+e_4([x])]\})\,|\, [x]\in \Xx\}\,,
\end{equation*}
which we also denote as $([e_1],\dots,[e_4],[e_1+\ldots+e_4])$.
\vs

\begin{remark}
Five projective vector fields are required to uniquely specify a projective frame on the four-dimensional model space $\Xx$, based on the following argument~\cite{Gallier2011}.
Assume that $[e_\mu] = [f_\mu]$ for all $\mu=1,\dots, 4$. 
Then $e_\mu = \lambda_\mu f_\mu$ for $\lambda_\mu \in C^\infty(\Xx,\rr_{\neq 0})$. 
With the fifth projective vector field $[e_1+\ldots+e_4] = [f_1+\ldots +f_4]$ we have $e_1+\ldots+e_4 = \lambda (f_1+\ldots+f_4)$ for some $\lambda\in C^\infty(\Xx,\rr_{\neq 0})$. 
Together this yields $0=(\lambda_1 -\lambda)f_1+\ldots+(\lambda_4-\lambda)f_4$, such that by linear independence of the $f_\mu$: $\lambda_\mu = \lambda$ for all $\mu = 1,\dots,4$. 
In the projective setting, the two frames agree. 
\end{remark}
\vs

The following example shows that projective vector fields decompose non-uniquely into projective linear combinations of representatives of projective frame vector fields, which is different from non-projective vector fields on $\Xx$.
\vs

\begin{example}\label{ExampleVectorFieldTrafo}
Consider Poincaré geometry $\Gg((1),(3,1))$.
The map $[x_0,x_1,\dots,\allowbreak x_4]\mapsto (x_1/x_0,\dots,x_4/x_0)$ identifies $\Xx((1),(3,1))$ with $\Abb^{3,1}$.
With $y_\mu:= x_\mu/x_0$, $\mu=1,\dots,4$, the standard vector fields on $\Abb^{3,1}$ are given by $\partial/\partial y_\mu = x_0\, \partial/\partial x_\mu$.
A projective vector field $[W]$ on $\Xx((1),(3,1))$ can be written as a linear combination
\begin{equation*}
[W([x])] = \bigg[\sum_{\mu=1}^4 W_\mu([x])\, \frac{\partial}{\partial y_\mu}  + W_5([x])  \sum_{\mu=1}^4 \frac{\partial}{\partial y_\mu}\bigg]\,.
\end{equation*}
The functions $W_1,\dots,W_5$ are unique up to common prefactors as well as functions $f\in C^\infty(\Xx((1),(3,1)))$ acting as $W_1,\dots,W_5 \mapsto W_1-f,\dots,W_4-f, W_5+f$.

Still, we write the (non-unique) projective components of $[W]$ as $[W_1,\dots,W_5]$, when we can make well-posed statements about them.
For instance, the components $[W_1,\dots,W_5]$ uniquely define the set of homogeneous coordinates of $[W]$ with respect to $([\partial/\partial_1],\dots,[\partial/\partial y_4],[\partial/\partial y_1+\ldots + \partial/\partial_4])$, which is
\begin{align*}
&\bigg\{(\lambda W_1',\dots,\lambda W_4')\in C^\infty(\Xx)^4\,\bigg|\, W_\mu'\in C^\infty(\Xx), \forall [x]\in\Xx\, \exists \,\mu\in\{1,\dots,4\}: W_\mu'([x])\neq 0,\nonumber\\
&\qquad\qquad\qquad\qquad\qquad\qquad\qquad \lambda\in C^\infty(\Xx,\rr_{\neq 0}), \bigg[\sum_{\nu=1}^4 W_\nu' \partial/\partial y_\nu\bigg] = [W]\bigg\}\,.
\end{align*}
This assumes $[W]\neq 0$.
\end{example}
\vs

To study the behavior of a projective vector field $[W]$ under deformations and limits of geometries, let $([e_1],\dots,[e_5])$ be a projective frame on $\Xx$ and $[W_1,\dots,W_5]$ be the components of $[W]$ with respect to $([e_1],\dots,[e_5])$.
A sequence $[b_n]\in\Pp\GL_5\rr$ acts on the components $W_A$ as
\begin{equation}\label{EqVectorFieldDeformationMap}
([x], [W_1([x]),\dots,W_5([x])])\mapsto ([b_n\cdot x], [b_n]\cdot [W_1([x]),\dots,W_5([x])])\,.
\end{equation}
If $\lim_{n\to\infty}[b_n]$ exists as a projective $5\times 5$ matrix, the $n\to\infty$ limit of \eqref{EqVectorFieldDeformationMap} is well-defined, since no infinite blow-ups of components appear in the projective formulation. 
Yet, the limit matrices $\lim_{n\to\infty} [b_n]$ can decrease in rank, see e.g.~the $n\to \infty$ limit of the matrix \eqref{EqNonRelLimitMatrices}, whose rank decreases from 5 to 3. 
Projective vector fields can thus degenerate in limits, such that they locally provide elements of certain strict subspaces of the projective tangent spaces of $\Xx$.

We define projective tensor fields on $\Xx$ analogously to projective vector fields modulo multiplication by $C^\infty(\Xx,\rr_{\neq 0})$-elements.
On projective tensor fields such as projective metrics or projective differential forms, deformations of geometries act as tensor contractions with the $[b_n]$ or $[b_n^{-1}]$, depending on the type of the projective tensor field under consideration. 
Analogously to projective vector fields, projective tensor fields can degenerate in limit processes, but remain well-defined if $\lim_{n\to\infty}[b_n]$ exists as a projective $5\times 5$ matrix.
For instance, no singularities appear for projective metrics, which can solely degenerate.%
\footnote{This is in contrast to standard general relativity, where locally singular behavior of metrics can render the theory inconsistent in their surrounding neighborhoods.}
\vs

\begin{remark}
Projective equivalence classes of tensor fields are similar to conformally equivalent tensor fields.
If a tensor field $Q$ on $\RP^4$ has non-zero conformal weight $\Delta$, its projective equivalence class $[Q]$ can be identified with the set generated by its Weyl transformations $Q\mapsto \exp(\Delta \omega)Q$ for $\omega\in C^\infty(\RP^4)$ together with the reflection $Q\mapsto - Q$.
\end{remark}


\section{Quantum fields on projective geometries}\label{SecQFTOnDeformedGeometries}
Quantum fields are tied to space-time geometries by means of their behavior under space-time symmetry transformations, which involves projective unitary representations of the space-time symmetry groups acting on the Hilbert space~\cite{Weinberg:1995mt}.
Irreducibility of the representations characterizes fundamental quantum particles, and the Hilbert space of the theory is constructed from these. 
Geometries thus dictate the types of particles that can occur and at least partly the structure of the Hilbert space.
For instance, on Poincaré geometry particles are described by the well-known massive spin and massless helicity representations of the Poincaré group.

\subsection{Projective quantum fields}\label{SecQuantumFields}
We define projective quantum fields similarly to Wightman quantum fields~\cite{Haag1996Local}, thereby in a certain way generalizing the projective tensor fields considered in \Cref{SecProjTensorFields}.
We do not aim for their complete mathematical characterization, but restrict to the representa-tion-theoretic assumptions necessary for our geometry-related derivations. 
Let $\Hcal$ be a complex Hilbert space with inner product $\langle\cdot ,\cdot \rangle$ (anti-linear with respect to its first argument), on which a projective quantum field can be defined.
$\Ucal(\Hcal)$ denotes the set of unitary linear operators on $\Hcal$ and $\overline{\PGL_5\rr}$ the universal cover of $\PGL_5\rr$. 
\vs

\begin{defi}\label{DefFieldOperator}
A \emph{projective quantum field} is a tuple $\hat{\Ocal} = (U, \rho, \{[\hat{\Ocal}([x])]\,|\,[x]\in \RP^4\})$ consisting of:
\begin{enumerate}[(i)]
\item a projective unitary $\PGL_5\rr$ representation $U:\PGL_5\rr \to \Ucal(\Hcal)$ with adjoint ${}^\dagger$,
\item a finite-dimensional complex $\overline{\PGL_5\rr}$ representation $\rho$ with $\dim\rho\neq 0$, and
\item equivalence classes of tuples of non-zero linear operator-valued tempered distributions $[\hat{\Ocal}([x])] = [\hat{\Ocal}_1([x]),\dots,\hat{\Ocal}_{\dim\rho}([x])]$, $[x]\in\RP^4$, defined modulo $C^\infty(\RP^4,\allowbreak \rr_{\neq 0})$ prefactors, where all $\hat{\Ocal}_\alpha([x])$ have a common dense domain of definition $\Dcal\subset\Hcal$, $U([g])\Dcal \subset \Dcal$ and $\hat{\Ocal}_\alpha([x])\Dcal\subset\Dcal$ for all $[g]\in \PGL_5\rr$, $[x]\in\RP^4,\alpha=1,\dots,\dim\rho$, and for which the smeared representatives
\begin{equation*}
\hat{\Ocal}_\alpha(f,\RP^4):= \int_{\RP^4} \hat{\Ocal}_\alpha([x])\, f([x])\, \dd^4[x]
\end{equation*}
are bounded operators on $\Dcal$ for all $f\in C^\infty(\RP^4)$,
\end{enumerate}
such that for all $[x]\in\RP^4$, $[g]\in \PGL_5\rr$, $\alpha=1,\dots,\dim\rho$:
\begin{equation}\label{EqGlobalCovariance}
U([g])\hat{\Ocal}_\alpha([x])U^\dagger([g]) = \sum_{\beta=1}^{\dim\rho}\rho_{\alpha\beta}([g^{-1}]) \hat{\Ocal}_\beta([g\cdot x])\,.
\end{equation}
Often we omit the component indices $\alpha,\beta$ from notations.
When denoted as $\hat{\Ocal}$, $\hat{\Ocal} = (U, \rho, \{[\hat{\Ocal}([x])]\,|\,[x]\in\RP^4\})$ is understood.
The $[\hat{\Ocal}([x])]$ are called \emph{field operators}.
On the geometry $(\Xx,\Oo) <  (\RP^4,\PGL_5\rr)$ the projective quantum field is given by restriction:
\begin{equation*}
\hat{\Ocal}|_{(\Xx,\Oo)}:= (U|_\Oo,\rho,\{[\hat{\Ocal}([x])]\,|\, [x]\in\Xx\})\,.
\end{equation*}
We also call $\hat{\Ocal}|_{(\Xx,\Oo)}$ a \emph{restricted} projective quantum field.
It comes with smeared representatives
\begin{equation*}
\hat{\Ocal}_\alpha(f,\Xx) := \int_\Xx \hat{\Ocal}_\alpha([x])\,f([x])\, \dd^4[x]\,,
\end{equation*}
where $f\in C^\infty(\Xx)$.
\end{defi}
\vs

The smeared representatives are defined not for Schwartz functions, but for smooth functions, since $\RP^4$ is compact.
We note that the finite-dimensional representations of $\overline{\PGL_5\rr}$ all factor through representations of $\PGL_5\rr$ and are thus not faithful.
Upon restriction to the geometry $(\Xx,\Oo)$, no restriction of the representation $\rho$ to $\Oo$ is included, so that local changes of (projective) reference frames can be consistently described in an accompanying work~\cite{spitz2024similarities}.
No equal-time commutation relations and causality preservation are assumed, whose formulation for curved geometries would rest on an observer-specific, distinguished time direction as available e.g.~for globally hyperbolic space-times~\cite{oneill1983semiriemannian}.
As part of an algebraic characterization of projective quantum fields, \Cref{DefFieldOperator} could be incorporated in Haag-Kastler conditions for nets of von Neumann operator algebras~\cite{Haag1996Local, Borchers:2000pv, Witten:2018zxz}, if these are formulated for the ambient projective geometry $(\RP^4,\PGL_5\rr)$. 
\vs

\begin{remark}
\Cref{DefFieldOperator} rests upon $(\RP^4,\PGL_5\rr)$ as the ambient geometry.
This choice is at least in parts without loss of generality, see \Cref{AppendixAmbientGeometryChoice}.
\end{remark}
\vs

\begin{remark}
While $U$ is a projective unitary representation in \Cref{DefFieldOperator}, $\rho$ is non-projective. 
This is sensible, since if $\rho$ was projective, its multiplier would be equivalent to the trivial one, see \Cref{PropRhoHasTrivialMultiplier} in \Cref{AppendixAspectsProjectivity}.
The finite dimensionality of $\rho$ is part of established formal approaches~\cite{Haag1996Local}.
\end{remark}
\vs

\Cref{DefFieldOperator} indeed extends the unitary behavior of quantum fields on fixed geometries, as the following example demonstrates for Poincaré geometry.
\vs

\begin{example}\label{ExampleVectorPoincareRep}
Consider a projective quantum field $(U,\CP^4_{\PGL_5\rr},\{[\hat{\Ocal}([x])]\,|\,[x]\in\RP^4\})$, where $\CP^4_{\PGL_5\rr}$ denotes the fundamental complex representation of $\PGL_5\rr$ acting on $\CP^4$ by (projective) matrix multiplication.
$\CP^4_{\PGL_5\rr}$ provides a $\overline{\PGL_5\rr}$ representation by previously applying the cover projection $\overline{\PGL_5\rr}\to \PGL_5\rr$, which we often omit from notations.
This describes projective quantum fields, which resemble projective vector fields, see \Cref{SecProjTensorFields}.
Restricted to Poincaré geometry $\Gg((1),(3,1))$, the action of $U|_{\PO((1),(3,1))}$ is to be considered.
The generalized unitary transformation behavior \eqref{EqGlobalCovariance} of the field operators $[\hat{\Ocal}([x])]$ yields for a Poincaré transformation $[(\Lambda,t)]\in\PO((1),(3,1))$ (see \Cref{EqPoincareTrafo55}):
\begin{equation*}
U([(\Lambda,t)]) [\hat{\Ocal}([x])] U^\dagger([(\Lambda,t)]) = [(\Lambda,t)]^{-1}\cdot [\hat{\Ocal}([\Lambda\cdot x + t])]\,,
\end{equation*}
where $[(\Lambda,t)]^{-1} = [(\Lambda^{-1},-\Lambda^{-1}\cdot t)]$.
If $\hat{\Ocal}_1([x]) = 0$, we have for the other components of representative field operators $\hat{\Ocal}([x])$ of $[\hat{\Ocal}([x])] = [\hat{\Ocal}_1([x]),\dots,\hat{\Ocal}_5([x])]$:
\begin{equation}\label{EqLorentzTrafoVectorFieldBehavior}
U([(\Lambda,t)]) \hat{\Ocal}_{\mu+1}([x]) U^\dagger([(\Lambda,t)]) = \sum_{\nu=1}^4 (\Lambda^{-1})_{\mu\nu} \hat{\Ocal}_{\nu +1}([\Lambda\cdot  x + t])
\end{equation}
for all $\mu=1,\dots,4$.
This is the usual behavior of vector quantum fields under Poincaré transformations.
\end{example}
\vs

\begin{remark}
The massless, infinite spin representations of the Poincaré group~\cite{Weinberg:1995mt} cannot appear upon restriction of a projective quantum field to Poincaré geometry, since the representation $\rho$ is required to be finite-dimensional.
\end{remark}
\vs

Projective quantum fields transform naturally and smoothly under geometry deformations up to action of the finite-dimensional representation $\rho$ as specified by the following lemma.
Given a model space $\Xx$ and a sequence $[b_n]\in\PGL_5\rr$, for the treatment of geometry limits we denote the set of point-wise limits of the form $\lim_{n\to\infty} [b_n\cdot x]$ for $[x]\in \Xx$ by $\lim_{n\to\infty}[b_n]\cdot \Xx$, which in general is different from limits of model spaces with respect to the model space topology.
In fact, we show in the next subsection that the spaces $\lim_{n\to\infty}[b_n]\cdot \Xx$ can form dimensionally reduced subspaces of model space closures, see in particular \Cref{PropPointwiseConvergenceModelSpacePoints}.
\vs

\begin{lem}\label{LemmaQuantumFieldsTransformNaturallySmoothly}
Consider a projective quantum field $\hat{\Ocal}=(U,\rho,\{[\hat{\Ocal}([x])]\,|\,[x]\in\RP^4\})$ and the geometry $\Gg=(\Xx,\Oo)$, let $[b_n]\in\PGL_5\rr$ and $(\Xx',\Oo')=\lim_{n\to\infty} ([b_n]_* \Gg)$ be the limit geometry.
Assume $[b_\infty]:=\lim_{n\to\infty}[b_n]$ exists as a projective $5\times 5$ matrix.
The action of $[b_n]$ on $\hat{\Ocal}|_{(\Xx,\Oo)}$ yields the restricted projective quantum field
\begin{equation*}
(U|_{\Ad_{[b_n]\Oo}},\rho,\{[\rho([b_n^{-1}])\hat{\Ocal}([x])]\,|\,[x]\in[b_n]\cdot \Xx\})
\end{equation*}
with $n\to\infty$ limit
\begin{equation*}
\left(U|_{\Oo'},\rho,\left\{\left.[\rho_\infty]\cdot [\hat{\Ocal}([x])]\,\right|\, [x]\in\lim_{n\to\infty}[b_n]\cdot \Xx\right\}\right)\,,
\end{equation*}
where $[\rho_\infty] = \lim_{n\to \infty}[\rho([b_n^{-1}])]\in \Pp\mathrm{Mat}_{\dim\rho\times\dim\rho}\cc$, which is the space of projective, complex $\dim\rho\times\dim\rho$ matrices.
Further, $U([g]) [\hat{\Ocal}([x])] U^\dagger([g])$, $[\hat{\Ocal}([g\cdot x])]$ and all their representatives depend smoothly on $[g]\in\PGL_5\rr$ and $[x]\in\RP^4$, where convergence can be defined with regard to the operator norm on $\Dcal$ for the smeared operator variants.
\end{lem}

\begin{proof}
Let $\omega$ be the (Schur) multiplier of the projective representation $U$.
A geometry deformation via $[b_n]\in\PGL_5\rr$ acts on representatives $\hat{\Ocal}([x])$ of the field operators $[\hat{\Ocal}([x])]$, transformed via $[g]\in\Oo$, as
\begin{align}
U([g]) \hat{\Ocal}([x])  U^\dagger([g])\mapsto &\; U([b_n]) U([g])\hat{\Ocal}([x]) U^\dagger([g]) U^\dagger([b_n]) \nonumber\\
& =\omega([b_n],[g]) \omega ( [b_ng],[b_n^{-1}])  U([b_n gb_n^{-1}])  U([b_n]) \hat{\Ocal}([x]) U^\dagger([b_n])\nonumber\\
&\qquad\qquad\qquad \times \omega([b_n],[g^{-1}]) \omega([b_ng^{-1}],[b_n^{-1}])  U^\dagger([b_n g b_n^{-1}])\nonumber\\
& = U([b_n g b_n^{-1}]) U([b_n]) \hat{\Ocal}([x]) U^\dagger([b_n]) U^\dagger([b_n g b_n^{-1}])\,.\label{EqUOUDecompMultiplier}
\end{align}
Here we used multiplier cyclicity, i.e.,
\begin{equation*}
\omega([b_n],[g]) \omega([b_n g],[b_n^{-1}]) = \omega([b_n],[gb_n^{-1}])\omega([g],[b_n^{-1}])\,,
\end{equation*}
and
\begin{align*}
\omega([b_n],[g^{-1}]) U([b_n g^{-1}]) =&\;  U([b_n]) U([g^{-1}]) \nonumber\\
= &\; (U([g]) U([b_n^{-1}]))^{-1} = \frac{1}{\omega([g],[b_n^{-1}])} U([b_n g^{-1}])\,,
\end{align*}
analogously for $\omega([b_ng^{-1}],[b_n^{-1}])$.
By means of \Cref{EqUOUDecompMultiplier}, the symmetry group $\Oo$ is conjugated by $[b_n]$ and the field operator $\hat{\Ocal}([x])$ is changed to 
\begin{equation*}
U([b_n])\hat{\Ocal}([x])U^\dagger([b_n]) = \rho([b_n^{-1}]) \hat{\Ocal}([b_n\cdot x])\,.
\end{equation*}

We prove in \Cref{AppendixCommutativeDiagramReps} that the conjugacy limit of $U(\Oo)$ in $U(\PGL_5\rr)$ via $U([b_n])$ is the same as $U(\Oo')$ up to multiplier, i.e., the diagram
\begin{equation*}
\begin{tikzcd}
\lim_{n\to\infty} \Ad_{[b_n]}: &  \Oo\arrow{r}\arrow{d}{U|_{\Oo}} &  \Oo' \arrow{d}{ [ U|_{\Oo'} ] }\\
\lim_{n\to\infty} [\Ad_{U([b_n])} ] : & U(\Oo) \arrow{r} & \left[U(\Oo')\right]
\end{tikzcd}
\end{equation*}
commutes. 
Here, $[U(\Oo')]$ denotes $U(\Oo')$ modulo $\Uu(1)$ prefactors, analogously for the maps in square brackets.
In fact, by the considerations of \Cref{AppendixCommutativeDiagramReps} multipliers of $U|_{\Oo'}$ are limits of multipliers of $U|_{\Oo}$, which can become trivial.
The matrix $[\rho_\infty]$ exists, since $[b_\infty]$ exists and $\rho([b_n^{-1}])$ can be viewed as a sequence in the compact space $\RP^{(\dim\rho)^2-1}$.
This shows that restricted projective quantum fields behave in the geometry limit as claimed.

Concerning smoothness, the generalized unitary transformation behavior \eqref{EqGlobalCovariance} yields
\begin{equation*}
[\hat{\Ocal}([g\cdot x])] = \Ad_{U([g])} [\rho([g]) \hat{\Ocal}([x])]\,.
\end{equation*}
By the definition of Lie group representations, $[\rho([g])]$ and $U([g])$ depend smoothly on $[g]$, thus also $\Ad_{U([g])}$.
Therefore, $[\hat{\Ocal}([g\cdot x])]$ and $U([g]) [\hat{\Ocal}([x])] U^\dagger([g])$ depend smoothly on $[g]$, where convergence can be defined with regard to the norm on $\Dcal$ for representatives of the smeared operator variants.
Smooth dependence on $[x]\in\RP^4$ follows from transitivity of the smooth group action.
\end{proof}

\subsection{Projective correlators}\label{SecProjCorrelators}
Given a four-dimensional geometry $(\Xx,\Oo)$, an algebra $\Afrak(\Xx)$ of projective correlators can be defined as
\begin{equation*}
\Afrak(\Xx):=\bigg\{\big[\hat{\Ocal}^{(\dagger)}(f_1,\Xx)\otimes \ldots \otimes \hat{\Ocal}^{(\dagger)}(f_\ell,\Xx)\big]\bigg|\, f_i\in C^\infty(\RP^4), \ell\in\nn\bigg\}\,,
\end{equation*}
where the involved $\hat{\Ocal}([x])$ are representatives of $[\hat{\Ocal}([x])]$ and the superscript $(\dagger)$ denotes for each tensor product factor individually, whether taking the adjoint or not.
The tensor products are defined with respect to the components $\hat{\Ocal}_\alpha(f_i,\Xx)$ of the smeared field operator representatives.
The involved functions $f_i$ are defined on $\RP^4$, so that the action of geometry deformations and limits of projective correlators can be consistently described.
We denote the analogous algebra constructed on the ambient geometry model space $\RP^4$ by $\Afrak(\RP^4)$.
Elements of $\Afrak(\Xx)$ and $\Afrak(\RP^4)$ form bounded operators on $\Dcal$, since they involve finitely many concatenations of operators, which are bounded on $\Dcal$.

Projective correlator algebras behave as follows under geometry deformations.
\vs

\begin{prop}\label{PropCanonicalTensorProductRepOperatorProducts}
Let $[g]\in\PGL_5\rr$ and $(\Xx,\Oo)< (\RP^4,\PGL_5\rr)$ be a geometry.
Deformations of $(\Xx,\Oo)$ via $[g]$ give the commutative diagram
\begin{equation*}
\begin{tikzcd}
 \Ad_{[g]}: &  \Oo\arrow{r}{\sim}\arrow{d}{\Ad_{U}} &  \Ad_{[g]}\Oo  \arrow{d}{\Ad_U}\\
 \Ad_{U([g])}: & \Afrak(\Xx) \arrow{r}{\sim} & \Afrak([g]\cdot \Xx)\,,
\end{tikzcd}
\end{equation*}
where the horizontal maps are isomorphisms of groups (top) and algebras (bottom).
The vertical map $\Ad_{U}$ acts for an element $[h]\in\Oo$ as $\Ad_{U([h])}$:
\begin{align}
&[\hat{\Ocal}^{(\dagger)}(f_1,\Xx)\otimes \ldots \otimes \hat{\Ocal}^{(\dagger)}(f_\ell,\Xx)]\nonumber\\
&\mapsto [(\rho([h^{-1}])\hat{\Ocal})^{(\dagger)}(f_1([h^{-1}\cdot]),[h]\cdot \Xx)\otimes \ldots \otimes (\rho([h^{-1}])\hat{\Ocal})^{(\dagger)}(f_\ell([h^{-1}\cdot]),[h]\cdot \Xx)]\,,\label{EqDeformationCorrelator}
\end{align}
and analogously horizontally for more general $[g]\in\PGL_5\rr$ instead of such $[h]$.
\end{prop}

\begin{proof}
Let $\hat{\Ocal}(f_i,\Xx)$ be representatives of the smeared field operators $[\hat{\Ocal}(f_i,\Xx)]$, $f_i\in C^\infty(\RP^4)$.
$\Ad_{U([g])}$ maps
\begin{equation*}
\hat{\Ocal}_{\beta_1}^{(\dagger)}(f_1,\Xx)\circ \cdots \circ\hat{\Ocal}_{\beta_\ell}^{(\dagger)}(f_\ell,\Xx)
\end{equation*}
to 
\begin{align}
&U([g])\hat{\Ocal}_{\beta_1}^{(\dagger)}(f_1,\Xx)\circ \cdots \circ \hat{\Ocal}_{\beta_\ell}^{(\dagger)}(f_\ell,\Xx) U^\dagger([g])\nonumber\\
&= \int_{\Xx^\ell} f_1([x_1])\cdots f_\ell([x_\ell]) \, \rho_{\beta_1\gamma_1}^{(\dagger)}([g^{-1}])\hat{\Ocal}^{(\dagger)}_{\gamma_1}([g\cdot x_1])\circ  \cdots\nonumber\\
&\qquad\qquad\qquad\qquad\qquad\qquad \circ\rho_{\beta_\ell\gamma_\ell}^{(\dagger)}([g^{-1}])\hat{\Ocal}^{(\dagger)}_{\gamma_\ell}([g\cdot x_\ell])\,\thisd^4 [x_1] \cdots \thisd^4 [x_\ell]\nonumber\\
& =\int_{([g]\cdot \Xx)^\ell} f_1([g^{-1} \cdot y_1])\cdots f_\ell([g^{-1} \cdot y_\ell]) \rho_{\beta_1\gamma_1}^{(\dagger)}([g^{-1}])\hat{\Ocal}^{(\dagger)}_{\gamma_1}([y_1])\circ  \cdots\nonumber\\
&\qquad\qquad\qquad\qquad\qquad\qquad\qquad \circ \rho_{\beta_\ell\gamma_\ell}^{(\dagger)}([g^{-1}])\hat{\Ocal}_{\gamma_\ell}^{(\dagger)}([y_\ell])\,\thisd^4 [y_1] \cdots \thisd^4 [y_\ell]\nonumber\\
& = (\rho([g^{-1}])\hat{\Ocal})^{(\dagger)}_{\beta_1}(f_1([g^{-1}\cdot]),[g]\cdot \Xx)\cdots(\rho([g^{-1}])\hat{\Ocal})^{(\dagger)}_{\beta_\ell}(f_\ell([g^{-1}\cdot]),[g]\cdot \Xx)\,,\label{EqTrafoDeformationCorrelator}
\end{align}
where $\circ$ denotes Hilbert space operator composition, and integration variables have been changed.
We used that Jacobi determinants are trivial in the projective setting.
The functions $f_i([g^{-1}\cdot ])$ are again in $C^\infty(\RP^4)$.
Therefore, $\Ad_{[U([g])]}$ provides an algebra isomorphism onto $\Afrak([g]\cdot X)$.
\end{proof}

Replacing the deformation via $[g]$ by an arbitrary sequence $[b_n]\in \PGL_5\rr$, \Cref{EqDeformationCorrelator} does not apply to its $n\to\infty$ limit, since integration variables cannot be changed for non-invertible $[b_\infty]$.
Instead, as indicated already for projective quantum fields, the support of projective correlators can shift to model space boundaries and other lower-dimensional model space subspaces in the $n\to\infty$ limit, and the projective correlators can degenerate with respect to their field operator representative components.
\vs

\begin{thm}\label{ThmCorrelatorDegenerationInLimits}
Let $[\hat{C}]\in\Afrak(\Xx((p_0,q_0),\ldots,(p_k,q_k))$ for $p_i,q_i\in \nn$ with $\sum_{i=0}^k p_i+q_i=5$ and $[b_n]\in\PGL_5\rr$, so that $\lim_{n\to\infty}[b_n]_*\Gg((p_0,q_0),\ldots,(p_k,q_k))=(\Xx',\Oo')$ is a well-defined geometry limit and $[b_\infty]$ exists as a projective $5\times 5$ matrix.
Then $\lim_{n\to\infty} [b_n]_* [\hat{C}]$ is a well-defined, bounded operator on $\Dcal$.
If $[b_\infty]$ has not full rank, $\lim_{n\to\infty} [b_n]_* [\hat{C}]$ is degenerate as a finite-rank tensor with field operator components and has model space support within $\lim_{n\to\infty}[b_n]\cdot \Xx$, which is generally the union of a subspace of the boundary $\partial \overline{\Xx'}$ and another subspace of $\Xx'$ of dimension strictly smaller than $4$.
\end{thm}
\vs

\Cref{ThmCorrelatorDegenerationInLimits} is analogous to the behavior of projective tensor fields at individual space-time points under geometry deformations and limits, see e.g.~\Cref{EqVectorFieldDeformationMap}.
Before proving it, we show an auxiliary proposition regarding the point-wise convergence of certain sequences of points within model spaces, which form $\lim_{n\to\infty}[b_n]\cdot \Xx$.
\vs

\begin{prop}\label{PropPointwiseConvergenceModelSpacePoints}
Let $[x]\in\Xx((p_0,q_0),\ldots,(p_k,q_k))$, consider a sequence $[b_n]\in\PGL_5\rr$ and the geometry limit $(\Xx((p_0,q_0),\ldots,(p_k,q_k)),\PO((p_0,q_0),\ldots,(p_k,q_k)))\to (\Xx',\Oo')$.
Assume $[b_\infty]$ does not have full rank.
Then the $n\to \infty$ limit of $[b_n\cdot x]$ is either in $\partial\overline{\Xx'}$, where the boundary and the closure are defined with respect to point-wise convergence, or in another strictly lower-dimensional submanifold of $\Xx'$, which depends only on the sequence $[b_n]$.
\end{prop}

\begin{proof}
We provide the proof for $p_i,q_i=0$ for all $i\geq 1$.
The general case follows straight-forwardly with additional index book-keeping.
We employ the $KAK$ decomposition of $\PGL_5\rr$, where $K$ is the maximal compact subgroup $\PO(5)$ and $A$ is the subgroup of diagonal, invertible projective $5\times 5$ matrices, see e.g.~\cite{knapp2001representation}.
Hence, for all $n$ there exist $[k_n], [l_n]\in K, [a_n]\in B$, so that $[b_n] = [k_n a_n l_n]$, where $[a_n]$ can be chosen of the canonical form
\begin{equation}\label{EqanConvergenceBehavior}
a_{n,ii}/a_{n,(i+1)(i+1)}\to \begin{cases}
1 & \mathrm{for}\, i\notin I\,,\\
0 & \mathrm{for}\, i\in I\,,
\end{cases}
\end{equation}
as $n\to\infty$ for some index set $I\subset \{1,2,3,4\}$.
This choice can be made, since else we can multiply the sequence $[a_n]$ by another sequence $[a_n']\in A$, which remains in a compact subset of $A$ and can be incorporated by left multiplication of $[a_n]$ with $[a_n']$.
$I$ is non-empty, since $[b_\infty]$ has not full rank.
Let $i_1<i_2<\ldots < i_\ell$ be the elements of $I$ in increasing order.
The sequences $[k_n],[l_n]$ have subsequences $[k_{n_j}], [l_{n_j}]$, which converge to some $[k],[l]\in K$, respectively, since $K$ is compact.
We write 
\begin{equation*}
[\tilde{x}] := [l\cdot x] = [\tilde{z}_0,\tilde{z}_1,\ldots, \tilde{z}_{\ell-1},\tilde{z}_\ell]\,,
\end{equation*}
where $\tilde{z}_0 := (\tilde{x}_0,\ldots,\tilde{x}_{i_1-1})$, $\tilde{z}_j:=(\tilde{x}_{i_j},\ldots,\tilde{x}_{i_{j+1}-1})$ with $i_{\ell+1} =m$.
We set $j_{\min}:=\min \{j=0,\ldots,4\,|\, \tilde{z}_j\neq 0\}$ and $j_{\max}:=\max \{j=0,\ldots,4\,|\, \tilde{z}_j\neq 0\}$.
If $j_{\min} = j_{\max}$, the convergence behavior~\eqref{EqanConvergenceBehavior} of the sequence $[a_n]$ yields that $[a_n \tilde{x}]\to [\tilde{x}]$ as $n\to\infty$.
These points form a lower-dimensional submanifold of $\overline{\Xx'} \cap [l]\cdot \Xx(p,q)$.

The case $j_{\min} < j_{\max}$ yields $[a_n \tilde{x}]\to [0,\ldots,0,\tilde{z}_{j_{\max}},0,\ldots,0]$.
We note that similarly to its $KAK$ decomposition we can apply the $KBH$ decomposition of $\PGL_5\rr$~\cite{heckman1995harmonic}, so that $[b_n] = [k_n' b_n' h_n]$ for $[k_n']\in K=\PO(5)$, $[b_n'] \in B=A$, $[h_n] \in \PO(p,q)$.
Suppose that $[h_n]\in \PO(p,q)\cap \PO(5)$.
Then the $KAK$ and $KBH$ decompositions of $[b_n]$ can be chosen identically, so that $[k_n'] = [k_n]$, $[b_n'] = [a_n]$ and $[h_n] = [l_n]$.
In particular, 
\begin{equation}\label{EqLimitPOpqanspecialcase}
\lim_{n\to\infty} \Ad_{[a_n]}\Ad_{[l]}\PO(p,q) =  \lim_{n\to\infty} \Ad_{[a_n]} \PO(p,q) = \PO((p_0',q_0'),\ldots,(p_{\ell-1}',q_{\ell-1}'))
\end{equation}
for the integer pairs $(p_i',q_i')$ determined uniquely by splitting the sequence $-1,\ldots,-1,\allowbreak +1,\ldots,+1$ ($-1$ appearing $p$ and $+1$ appearing $q$ times) after $i_1$ elements, then after the next $i_2-i_1$ elements and so forth.
Counting the integers $-1$ in the $i$-th step yields $p_i'$ and counting $+1$ yields $q_i'$.
A point $[x']$ is in $\Xx((p_0',q_0'),\ldots,(p_{\ell-1}',q_{\ell-1}'))$, if and only if
\begin{equation*}
-(x'_0)^2-\ldots-(x'_{p_0'-1})^2+(x'_{p_0'})^2+\ldots+(x'_{p_0'+q_0'-1})^2<0\,.
\end{equation*}
For the point $[x'] = [0,\ldots,0,\tilde{z}_{j_{\max}},0,\ldots,0]$ the polynomial on the left-hand side evaluates to zero.
Thus, $\lim_{n\to\infty}[a_n \tilde{x}]\in \partial \overline{\Xx((p_0',q_0'),\ldots,(p_{\ell-1}',q_{\ell-1}'))}$ and $\lim_{n\to\infty}[b_n x] =[k]\lim_{n\to\infty}[a_n \tilde{x}]\in \partial\overline{\Xx'}$.

The case $h_n\in \PO(p,q)\setminus (\PO(5)\cap \PO(p,q))$ remains.
Still, we have then that $[k\cdot (0,\ldots,0,\tilde{z}_{j_{\max}},0,\ldots,0)]\in \overline{\Xx'}$ and such points form a submanifold of $\overline{\Xx'}$ of dimension strictly lower than 4.
This concludes the proof.
\end{proof}

We can now provide the proof of \Cref{ThmCorrelatorDegenerationInLimits}.

\begin{proof}[Proof of \Cref{ThmCorrelatorDegenerationInLimits}]
The projective correlator $[\hat{C}]$ is of the form
\begin{equation*}
[\hat{C}] = \big[\hat{\Ocal}^{(\dagger)}(f_1,\Xx)\otimes \ldots \otimes \hat{\Ocal}^{(\dagger)}(f_\ell,\Xx)\big]
\end{equation*}
for $f_i\in C^\infty(\RP^4)$.
We choose representatives $\hat{\Ocal}([x])$ of $[\hat{\Ocal}([x])]$ and let $v\in \Dcal$.
For all $[b_n]$ we have $||U([b_n])||_\Dcal=1$ by unitarity of $U$.
Thus, with $[\rho_\infty]=\lim_{n\to\infty}[\rho([b_n^{-1}])]$ and $\rho_\infty$ a representative of $[\rho_\infty]$, we have for all $\alpha_1,\ldots,\alpha_\ell$ and all $n$ the upper bound
\begin{align}
&\langle v,\int_{\Xx((p_0,q_0),\ldots,(p_k,q_k))}\dd^4[x_1] f_1([x_1])(\rho([b_n^{-1}])\hat{\Ocal}([b_n x_1]))_{\alpha_1}^{(\dagger)}\nonumber\\
&\qquad\qquad\qquad\qquad  \circ\ldots \circ \int_{\Xx((p_0,q_0),\ldots,(p_k,q_k))}\dd^4[x_\ell] f_1([x_\ell])(\rho([b_n^{-1}])\hat{\Ocal}([b_n x_\ell]))_{\alpha_\ell}^{(\dagger)}v\rangle\nonumber\\
&\qquad = \langle v, \Ad_{U([b_n])} (\hat{\Ocal}_{\alpha_1}^{(\dagger)}(f_1,\Xx((p_0,q_0),\ldots,(p_k,q_k))) \nonumber\\
&\qquad\qquad\qquad\qquad \circ\ldots\circ \hat{\Ocal}_{\alpha_\ell}^{(\dagger)}(f_1,\Xx((p_0,q_0),\ldots,(p_k,q_k)))v\rangle\nonumber\\
&\qquad \leq \prod_{i=1}^\ell ||\hat{\Ocal}_{\alpha_i}^{(\dagger)}(f_i,\Xx((p_0,q_0),\ldots,(p_k,q_k))||_\Dcal < \infty\,.\label{EqUpperBoundCorrelatorDeformation}
\end{align}
Therefore, the $n\to\infty$ limit is finite:
\begin{align*}
\lim_{n\to\infty} \langle v, \Ad_{U([b_n])}(\hat{\Ocal}^{(\dagger)}_{\alpha_1}(f_1,\Xx)\circ  \ldots \circ \hat{\Ocal}_{\alpha_\ell}^{(\dagger)}(f_\ell,\Xx))v\rangle < \infty\,,
\end{align*}
so that $\lim_{n\to\infty} [b_n]_*[\hat{C}]$ is indeed a well-defined and bounded operator on $\Dcal.$
If $[b_\infty]$ does not have full rank, the involved multiplication with $\rho_\infty$ is responsible for the degeneration of projective correlators with respect its field operator representative components.
\Cref{PropPointwiseConvergenceModelSpacePoints} implies $\lim_{n\to\infty}[b_n x_i]\in \partial \overline{\Xx'}$ or the inclusion in another lower-dimensional subspace of $\Xx'$.
\end{proof}

\Cref{ThmCorrelatorDegenerationInLimits} implies that projective correlation functions, i.e., expectation values of projective correlators, are well-defined for all states in $\Hcal$, as the following corollary shows.
\vs

\begin{cor}\label{CorCorrelatorsHaveWellDefinedExpec}
Given the setting of \Cref{ThmCorrelatorDegenerationInLimits}, expectation values $\langle v,[\hat{C}]v\rangle$ are well-defined for all $v\in\Dcal\subset \Hcal$ and can be uniquely extended to $\Hcal$.
The dependence of $\langle v, \Ad_{U([b_n])}[\hat{C}]v\rangle$ on $[b_n]$ is smooth for $v\in\Dcal$ and its $n\to\infty$ limit is well-defined.
Expectation values of the limiting projective correlator have a unique extension to states in $\Hcal$.
\end{cor}

\begin{proof}
Let $\hat{C}$ be a representative of
\begin{equation*}
[\hat{C}] = [\hat{\Ocal}^{(\dagger)}(f_1,\Xx)\otimes \ldots \otimes \hat{\Ocal}^{(\dagger)}(f_\ell,\Xx)]\in \Afrak(\Xx((p_0,q_0),\ldots,(p_k,q_k))
\end{equation*}
for $f_i\in C^\infty(\RP^4)$.
Consider a sequence $v_j\in \Dcal$, $v_j\neq 0$.
\Cref{EqUpperBoundCorrelatorDeformation} yields
\begin{equation}\label{EqProjCorrThmPfIntequality2}
\frac{\langle v_j, \hat{C}_{\alpha_1\dots \alpha_\ell} v_j\rangle}{||v_j||^2_\Dcal} \leq ||\hat{C}_{\alpha_1\dots\alpha_\ell}||_\Dcal < \infty\,.
\end{equation}
Assume $v_j\to v'\in \Hcal = \overline{\Dcal}$ as $j\to \infty$.
The right-hand side of \eqref{EqProjCorrThmPfIntequality2} is independent from $v_j$, such that the $j\to\infty$ limit of $\langle v_j, \hat{C}_{\alpha_1\dots\alpha_\ell} v_j\rangle/||v_j||^2$ exists by continuity of the operator $\hat{C}_{\alpha_1\dots\alpha_\ell}$ on $\Dcal$ and is finite for all $\alpha_1,\dots,\alpha_\ell$.
Further, $||v'||^2 = \langle v',v'\rangle < \infty$. 
In the projective setting, the division by $||v_j||^2$ is an identity map, such that $\langle v_j,[\hat{C}] v_j\rangle / ||v_j||^2 = \langle v_j, [\hat{C}] v_j\rangle$ is also well-defined in the $j\to\infty$ limit.
The same argument applies to the $j\to\infty$ limit of $\lim_{n\to\infty}\langle v_j,\Ad_{U([b_n])} [\hat{C}] v_j\rangle$, using \Cref{ThmCorrelatorDegenerationInLimits} and elements from its proof.

The smooth dependence of the expectation values of $\Ad_{U([b_n])}[\hat{C}]$ on $[b_n]$ for states in $\Dcal$ follows from \Cref{LemmaQuantumFieldsTransformNaturallySmoothly}.
\end{proof}

\begin{example}\label{ExampleProjCorrelatorsIRUVLimits}
Consider the projective quantum field $(U,\CP^4_{\PGL_5\rr},\{[\hat{\Ocal}([x])]\,|\,[x]\in\RP^4\})$ for the geometry $\Gg(4,1)$ and the sequence
\begin{equation*}
[b_n] = \Pp\left(\begin{matrix}
e^{-4n} & \\
 & e^n\cdot 1_{4\times 4}
\end{matrix}\right)\,.
\end{equation*}
By \Cref{ExampleGeometryDeformations}(ii), $[b_n]_*\Gg(4,1)\to \Gg((1),(3,1))$ as $n\to\infty$. 
Acting via $\Ad_{U([b_n])}$ on projective correlators, the projective matrix $[b_n^{-1}]$ appears, which has the $n\to\infty$ limit
\begin{equation*}
[\rho_\infty] = \lim_{n\to\infty} [b_n^{-1}] = \Pp\left(\begin{matrix}
1 & \\
 & 0_{4\times  4}
\end{matrix}\right)\,.
\end{equation*}
Therefore, due to \Cref{ThmCorrelatorDegenerationInLimits} and its proof, the $n\to\infty$ limit of the sequence $\Ad_{U([b_n])}[\hat{C}]$, $[\hat{C}]\in\Afrak(\Xx(4,1))$, can only depend non-trivially on the component $\hat{\Ocal}_1([x])$ and not on $\hat{\Ocal}_2([x]),\dots,\allowbreak \hat{\Ocal}_{5}([x])$.
The limiting projective correlator $\lim_{n\to\infty}\Ad_{U([b_n])}[\hat{C}]$ has support on finitely many points in
\begin{equation*}
\partial \overline{\Xx((1),(3,1))}\cup\{[1,0,\ldots,0]\} = \{[0,x_1,\dots,x_4]\in \RP^4\}\cup\{[1,0,\ldots,0]\}=\RP^3\cup \mathrm{pt}\,.
\end{equation*}
We denote the algebra of such limiting projective correlators by $\Afrak_\mathrm{deg}(\partial \overline{\Xx((1),(3,1))}\cup\{[1,0,\ldots,0]\})$ and leave implicit the field operator components, which can appear non-trivially.

Similarly, the non-relativistic (Galilei) limit of $\Gg((1),(3,1))$, anti-de Sitter geometry $\Gg(3,2)$ and representations other than $\rho =\CP^4_{\PGL_5\rr}$ can be considered.
This leads to the examples described in brevity in \Cref{FigLimitingAlgebras}, from which one can infer that projective correlators of projective de Sitter and projective anti-de Sitter geometry yield similar degenerate projective correlators in Poincaré and Galilei limits.
\end{example}
\vs

\tikzset{
    labl/.style={anchor=north, rotate=90, inner sep=1.mm}
}

\begin{figure}
\begin{center}
\small
\begin{equation*}
\begin{tikzcd}[	row sep=0ex, column sep =10ex,
				/tikz/column 1/.append style={anchor=base west},
				/tikz/column 2/.append style={anchor=base west},
				/tikz/column 3/.append style={anchor=base west}]
\textrm{Space-time geometry:} & \textbf{de Sitter} & \textbf{Poincaré} \\
{}&&\\
& \Gg(4,1)\arrow{r}{[b_n]_*\textrm{ of \eqref{EqDeSitterToPoincarecn}}} & \Gg((1),(3,1)) \\
{}&&\\
\textrm{Proj. correl. algebra:} & \fbox{$\Afrak(\Xx(4,1))$}\arrow{r} & \mlnode{$\Afrak_{\textrm{deg}}(\partial\overline{\Xx((1),(3,1))}$ \\ $\qquad\qquad\cup \{[1,0,\ldots,0]\})$}\\
\rho = \CP^4_{\PGL_5\rr}: && \hat{\Ocal}_1\\
\rho = R(\CP^4_{\PGL_5\rr}): &&\hat{\Ocal}_2,\ldots,\hat{\Ocal}_5\\
{}&&\\
{}&&\\
\textrm{Space-time geometry:} & \textbf{anti-de Sitter} & \textbf{Poincaré} \\
{}&&\\
& \Gg(3,2)\arrow{r}{\sigma^{-1} [b_n]_*\textrm{ of \eqref{EqAntiDeSitterToPoincarecn}}} & \Gg((1),(3,1))\\
{}&&&\\
\textrm{Proj. correl. algebra:} & \fbox{$\Afrak(\Xx(3,2))$}\arrow{r} & \mlnode{$\Afrak_{\textrm{deg}}(\partial\overline{\Xx((1),(3,1))}$\\ $\qquad\qquad\cup \{[1,0,\ldots,0]\})$}\\
\rho = \CP^4_{\PGL_5\rr}: &&\hat{\Ocal}_1\\
\rho = R(\CP^4_{\PGL_5\rr}): && \hat{\Ocal}_2,\ldots,\hat{\Ocal}_5 \\
{}&&\\
{}&&\\
\textrm{Space-time geometry:} & \textbf{Poincaré} & \textbf{Galilei}\\
{}&&\\
&  \Gg((1),(3,1))\arrow{r}{(\tau^{-1} [b_n])_*\textrm{ of \eqref{EqNonRelLimitMatrices}}} & \Gg((1),(1),(3))\\
\textrm{Proj. correl. algebra:} & \fbox{$\Afrak(\Xx((1),(3,1)))$}\arrow{r} & \fbox{$\Afrak_{\textrm{deg}}(i^0_p\vee i^+_p\setminus [0])$}\\
\rho = \CP^4_{\PGL_5\rr}: && \hat{\Ocal}_1,\hat{\Ocal}_2\\
\rho = R(\CP^4_{\PGL_5\rr}): && \hat{\Ocal}_3,\hat{\Ocal}_4,\hat{\Ocal}_5
\end{tikzcd}
\end{equation*}
\normalfont
\caption{Limiting projective correlator algebras for a selection of geometry limits of projective de Sitter geometry $\Gg(4,1)$ and projective anti-de Sitter geometry $\Gg(3,2)$.
Convergence of algebras is defined as element-wise convergence.
$\Afrak_\textrm{deg}$ denotes an algebra of projective correlators, which are degenerate with respect to their field operator components.
The rows starting with representations $\rho=\ldots$ denote the non-trivial dependence on field operator components of the degenerate projective correlator algebra.
$R(\CP^4_{\PGL_5\rr})$ denotes the action of $\overline{\PGL_5\rr}$ on $\CP^4$ by applying the cover projection $\overline{\PGL_5\rr}\to\PGL_5\rr$ and inverse right multiplication.
The space $i^0_p:=\{[0,0,x_2,x_3,x_4]\in \RP^4\}\cong \RP^2$ is projective spatial infinity, $i^+_p:=\{[x_0,x_1,0,0,0]\in \RP^4\,|\, x_0\neq 0\}= \rr$ is projective time-like infinity and $i_p^0 \vee i^+_p\setminus [0]$ indicates their wedge product at $[0]$ with $[0]$ removed.}\label{FigLimitingAlgebras}
\end{center}
\end{figure}

\Cref{CorCorrelatorsHaveWellDefinedExpec} has shown that projective correlation functions remain well-defined in geometry limits.
In particular, this applies to the ultraviolet and infrared limits of projective Poincaré geometry $\Gg((1),(3,1))$, as we show now.
Again, the model space identification $\Xx((1),(3,1)) = \Abb^{3,1}$ is via the diffeomorphism 
\begin{equation*}
\xi([x_0,x_1,\dots,x_4]):=(y_1,\dots,y_4)\,,
\end{equation*}
where $y_\mu := x_\mu / x_0$ for $\mu=1,\dots,4$.
The metric tensor $-\dd y_1^2 - \dd y_2^2 - \dd y_3^2 + \dd y_4^2$ on $\Xx((1),(3,1))$ depends only on the image of $\xi$. 
A physical scale transformation acts on $[x]\in\Xx((1),(3,1))$ as $\xi([x])\mapsto s \xi([x])$ for $s > 0$. 
Acting with $\xi^{-1}$, this is equivalent to
\begin{equation*}
[x_0,x_1,\dots,x_4]\mapsto [x_0, s x_1,\dots, s x_4]\,.
\end{equation*}
The $\PGL_5\rr$-element implementing the scale transformation is
\begin{equation*}
[d(s)] = \Pp\left(\begin{matrix}
1 & \\
& s \cdot 1_{4\times 4}
\end{matrix}\right)\,.
\end{equation*}
Poincaré geometry remains invariant under the deformation by $[d(s)]$ for all $s\neq 0$.
Let $[\hat{C}]\in \Afrak(\Xx((1),(3,1)))$ be a projective correlator on Poincaré geometry.
The one-parameter family $\{\Ad_{U([d(s)])}[\hat{C}]\}_s$ describes the behavior of $[\hat{C}]$ under scale transformations.
The limit $s\to 0$ shrinks physical length scales to zero and, considering $\Ad_{U([d(s)])}[\hat{C}]$, describes the ultraviolet limit of the projective correlator.
The limit $s\to \infty$ stretches physical length scales infinitely and describes the infrared limit of the projective correlator.
\vs

\begin{cor}\label{CorCorrelatorsIRUVLimits}
The expectation values of projective correlators in $\Afrak(\Xx((1),(3,1)))$ have well-defined, finite ultraviolet and infrared limits for states in $\Hcal$.
\end{cor}

\begin{proof}
The statement is a direct consequence of \Cref{ThmCorrelatorDegenerationInLimits} and \Cref{CorCorrelatorsHaveWellDefinedExpec}, along with the preceding considerations.
\end{proof}

\Cref{PropPointwiseConvergenceModelSpacePoints} and its proof actually yield that the model space support of projective correlators on $\Xx((1),(3,1))$ shrinks in the infrared limit to finitely many points in the union of the boundary $\partial\overline{\Xx((1),(3,1))} = \{[0,x_1,\ldots,x_4]\in \RP^4\} = \RP^3$ and the zero-dimensional subspace $\{[1,0,\ldots,0]\}\in \Xx((1),(3,1))$.
In the ultraviolet limit the support converges as well to points in $\partial \overline{\Xx((1),(3,1))}\cup \{[1,0,\ldots,0]\}$.
Therefore, projective correlators on Poincaré geometry behave in their infrared and ultraviolet limits as for a three-dimensional QFT.
This is qualitatively consistent with many quantum field theories, considering e.g.~the infinite temperature limit of thermal field theories~\cite{Appelquist:1981vg, Jourjine:1983hi}.
The relation to renormalizability of projective correlation functions is to be discussed in future work.

\subsection{Superselection sectors}\label{SecSuperSelectionSectors}
Inequivalent multipliers of the projective representation $U$ yield different superselection sectors for the projective quantum field $\hat{\Ocal}$.
Restricted to Poincaré geometry, this provides the usual bosonic and fermionic projective representations. 
This classification can be uniquely extended to the ambient geometry $(\RP^4,\PGL_5\rr)$, as we discuss now.

For this we explicitly describe the inequivalent multipliers.
$\PGL_5\rr$ is connected, so for any $[g],[h]\in\PGL_5\rr$ we can construct a path $\gamma([g],[h])$ as follows.
Polar decomposition yields $[g] = [u_g \exp(Y_g)]$ and $[h] = [u_h \exp(Y_h)]$ for unique $[u_g],[u_h]\in \PO(5)$ and unique symmetric $[Y_g],[Y_h]\in \pgl_5\rr$.
The group $\PO(5)\cong \SO(5)$ (since $-1_{5\times 5}\notin \SO(5)$) is connected and compact, so $\exp:\pofrak(5)\to \PO(5)$ is surjective.
Therefore, there exists $[X_h] \in \pofrak(5)$, such that $[u_h]=[\exp(X_h)]$.
We define the loop $\gamma([g],[h]):\rr\to \PO(5)\hookrightarrow \PGL_5\rr$ as
\begin{equation*}
\gamma([g],[h])(s):=[\exp(s X_h) u_g] \,.
\end{equation*}
We note that such one-parameter subgroups of $\PO(5)$ are periodic, so that $\gamma([g],[h])$ can be equally defined with domain a compact interval.
With this the possible inequivalent multipliers of $U$ are described by the following proposition.
\vs

\begin{prop}\label{PropCohomologyGeneratorsPGL5}
$H^2(\PGL_5\rr,\Uu(1))\cong \zz_2$, generated by the two inequivalent cocycles 
\begin{equation*}
\omega_+([g],[h])=+1
\end{equation*}
and
\begin{equation*}
\omega_-([g],[h])=\begin{cases}
+1\,, & \text{if $\gamma([g],[h])$ is contractible}\,,\\
-1\,, & \text{if $\gamma([g],[h])$ is not contractible}\,,
\end{cases}
\end{equation*}
for all $[g],[h]\in \PGL_5\rr$.
The multiplier $\omega_-$ is well-defined, since contractibility of $\gamma([g],[h])$ is independent from the choice of $[X_h]\in \pofrak(5)$ in its construction.
\end{prop}

\begin{proof}
$\PGL_5\rr\cong\SL_5\rr$ since $-1_{5\times 5}\notin \SL_5\rr$. 
Since $\SL_5\rr$ is perfect, the universal central extension yields the short exact sequence of groups~\cite{weibel1994introduction, roger1995extensions}
\begin{equation*}
1\to H_2(\PGL_5\rr,\zz)\cong \pi_1(\PGL_5\rr) \to \overline{\PGL_5\rr}\to \PGL_5\rr\to 1\,,
\end{equation*}
with $\overline{\PGL_5\rr}$ the universal cover of $\PGL_5\rr$.
The fundamental group is $\pi_1(\PGL_5\rr)\cong \pi_1(\SL_5\rr) \cong \zz_2$. 
With $H_1(\SL_5\rr,\zz)\cong 0$ since $\SL_5\rr$ is perfect, we have by the universal coefficient theorem 
\begin{equation*}
H^2(\PGL_5\rr,\Uu(1))\cong \mathrm{Hom}_\zz(H_2(\PGL_5\rr,\zz),\Uu(1))\cong\mathrm{Hom}_\zz(\zz_2,\Uu(1))\cong \zz_2\,.
\end{equation*}
Generators can be chosen of the claimed topological type, since the loops $\gamma$ generate $\pi_1(\PO(5))\cong \pi_1(\PGL_5\rr)$ and $\PO(5)$ is compact.

In order to prove the independence of contractibility from the choice of $[X_h]$, assume there exist $[X_h']\in \pofrak(5)$, such that $[u_h] = [\exp (X_h)] = [\exp(X_h')]$.
Let the loop $\gamma([g],[h])$ be defined as before using $[X_h]$ and $\gamma'([g],[h])$ be defined analogously but using $[X_h']$.
We construct a homotopy-equivalence $\gamma_t([g],[h]):\rr\to\PO(5)$, $t\in [0,1]$, between the two loops.
Set
\begin{equation*}
\gamma_t([g],[h])(s) = [\exp(s(t X_h' + (1-t)X_h))u_g] \,.
\end{equation*}
The map $\gamma_t([g],[h])(s)$ is continuous in $t$ and $s$, and $\gamma_0([g],[h]) = \gamma([g],[h])$, $\gamma_1([g],[h]) = \gamma'([g],[h])$.
Therefore, $\gamma_t([g],[h])$ defines a homotopy-equivalence.
\end{proof}

The multipliers of $U$ coincide with the multipliers of the restriction of $U$ to homogeneous Lorentz geometry structure groups, the Poincaré group or the full $\PGL_5\rr$, as the following lemma shows.
By a homogeneous Lorentz geometry we mean a deformation of de Sitter geometry, i.e., a geometry of the form $[g]_*\Gg(4,1)$ for some $[g]\in\PGL_5\rr$.
\vs

\begin{lem}\label{LemmaMultiplierRestrictionAgreement}
Let $\Oo<\PGL_5\rr$ be of type $\Oo=\Ad_{[g]} \PO(4,1)$ or $\Oo = \PO((1),(3,1))$ for some $[x]\in\RP^4$.
Its cohomology group $H^2(\Oo,\Uu(1))$ is generated by the $\omega_\pm$ of \Cref{PropCohomologyGeneratorsPGL5}, continuously deforming the loop $\gamma$ in the definition of $\omega_-$ into a subgroup $\PO(3)<\PO((1),(3,1))$ or $\Ad_{[g]}\Pp(\Oo(3)\times 1)<\Ad_{[g]} \PO(4,1)$.
If a projective unitary $\PGL_5\rr$ representation $U$ has multiplier $\omega_i$ upon restriction to $\Oo$, $i\in \{+,-\}$, then $U$ has the same multiplier $\omega_i$ on the full $\PGL_5\rr$, and vice versa.
\end{lem}
\begin{proof}
The second Lie algebra cohomology groups of $\Ad_{[g]}\pofrak(4,1)$ and $\pofrak((1),(3,1))$ are trivial,
\begin{equation*}
H^2(\Ad_{[g]}\pfrak\ofrak(4,1),\rr) \cong H^2(\pfrak\ofrak(4,1),\rr) \cong H^2(\pfrak\ofrak((1),(3,1)),\rr) \cong 0\,, 
\end{equation*}
such that 
\begin{align*}
&H^2(\Ad_{[g]}\PO(4,1),\Uu(1)) \cong H^2(\PO(4,1),\Uu(1)) \cong H^2(\PO((1),(3,1)),\Uu(1)) \nonumber\\
&\qquad \cong \pi_1(\Ad_{[g]}\PO(4,1))\cong \pi_1(\PO(4,1))\cong \pi_1(\PO((1),(3,1)))\cong \zz_2
\end{align*}
by the same arguments as in the proof of \Cref{PropCohomologyGeneratorsPGL5}.
The paths $\gamma$ included in the definition of $\omega_-$ can be continuously deformed into the subgroup $\PO(3)<\PO((1),(3,1))$ or $\Ad_{[g]}\Pp(\Oo(3)\times 1)< \Ad_{[g]}\PO(4,1)$, since any loop in $\PO(5)$ is homotopy-equivalent to some loop in $\PO(3)$ and vice versa.
Then the 2-cycles $\omega_\pm$, potentially defined with deformed loops $\gamma$, generate the cohomology groups $H^2(\Ad_{[g]}\PO(4,1),\Uu(1))$, $H^2(\PO((1),(3,1)),\Uu(1))$ and $H^2(\PGL_5\rr,\allowbreak \Uu(1))$, since these are entirely due to non-trivial fundamental groups.

Assume a projective unitary representation $V$ of $\Oo$ of any type as in the claim acts on the Hilbert space $\Hcal$ with multiplier $\omega_-$ and $U$ as in the claim has multiplier $\omega_+$.
There exist $[g],[h]\in \Oo$ such that $\gamma([g],[h])\subset \Oo$ is not contractible, i.e., for these $[g],[h]$ we have $\omega_-([g],[h])=-1$.
On the other hand, $\omega_+([g],[h])=+1$.
Thus, $V$ cannot be the restriction of $U$ to $\Oo$; the multipliers need to agree.
The same holds if $V$ has multiplier $\omega_+$ and $U$ has multiplier $\omega_-$.
The converse holds trivially, showing the claim.
\end{proof}

This allows for the consistent classification of projective quantum fields into bosonic and fermionic ones.
\vs

\begin{defi}
A projective quantum field $\hat{\Ocal}$ is \emph{bosonic} if $U$ has a multiplier equivalent to the trivial multiplier $\omega_+$ and \emph{fermionic} if $U$ has a multiplier equivalent to $\omega_-$.
\end{defi}

The Hilbert space $\Hcal$ does not split into a non-trivial direct sum $\Hcal^+\oplus \Hcal^-$, where $U([g])$ acting on $\Hcal^\pm$, $[g]\in\PGL_5\rr$ arbitrary, comes with multiplier $\omega_\pm$.
Indeed, even if $U$ acts on a Hilbert space of the form $\Hcal^+\oplus\Hcal^-$, it acts with the same multiplier on both $\Hcal^+$ and $\Hcal^-$.
Therefore, in the setting of projective quantum fields, the multipliers $\omega_\pm$ of $\PGL_5\rr$ indeed label different superselection sectors, justifying the previous definition.

\subsection{Composite and irreducible projective quantum fields}\label{SecSpinStatistics}
For the further classification of projective quantum fields, we construct composite and irreducible projective quantum fields.

\subsubsection{Composite projective quantum fields}
Let $\hat{\Ocal} = (U, \rho,\{[\hat{\Ocal}([x])]\,|\,[x]\in\RP^4\})$ be a projective quantum field.
Choose representatives $\hat{\Ocal}([x])$ of the field operators $[\hat{\Ocal}([x])]$ and let $\hat{\Ocal}_\alpha^*([x]):\Dcal^*\to\Dcal^*$ be the dual field operator to $\hat{\Ocal}_\alpha([x])$.
Denote by $\Xi$ the anti-linear, bijective Riesz map $\Dcal\to\Dcal^*, v\mapsto \langle v,\cdot\rangle$ with inverse $\Xi^{-1}: \varphi = \langle v,\cdot\rangle\mapsto v$.
We often omit the map $\Xi$ from notations.
\vs

\begin{defi}
We call a projective quantum field $(U',\rho',\{[\hat{\Ocal}'([x])]\,|\, [x]\in\RP^4\})$ a \emph{composite projective quantum field} of $\hat{\Ocal}$, if $U'=U$ and representative field operators $\hat{\Ocal}'([x])$ can be chosen, such that there exist $a_{\gamma,\alpha_1\dots\alpha_p,\beta_1\dots\beta_q}\in\cc$, $p,q\in \nn$:
\begin{align}\label{EqCompositeFieldOpDecomposition}
\hat{\Ocal}'_\gamma([x]) =&\; \sum_{\alpha_1,\dots,\alpha_p=1}^{\dim\rho} \sum_{\beta_1,\dots,\beta_q=1}^{\dim\rho} a_{\gamma,\alpha_1\dots\alpha_p,\beta_1\dots\beta_q} \hat{\Ocal}_{\alpha_1}([x])\circ\ldots \circ \hat{\Ocal}_{\alpha_p}([x]) \nonumber\\
&\qquad\qquad\qquad\qquad\qquad\qquad \circ \Xi^{-1}\circ \hat{\Ocal}^*_{\beta_1}([x])\circ\ldots\circ \hat{\Ocal}^*_{\beta_q}([x])\circ \Xi
\end{align}
for all $[x]\in\RP^4$, $\gamma=1,\dots,\dim\rho'$.
The $[\hat{\Ocal}'([x])]$ are called \emph{composite field operators}.
\end{defi}
\vs

Note that the coefficients $a_{\gamma,\alpha_1\dots\alpha_p,\beta_1\dots\beta_q}\in\cc$ in \Cref{EqCompositeFieldOpDecomposition} are vastly restricted by the generalized unitary transformation property \eqref{EqGlobalCovariance} of projective quantum fields, which both $\hat{\Ocal}$ and $\hat{\Ocal}'$ obey.
\vs

\begin{example}
Define field operator representatives as 
\begin{align*}
&(\hat{\Ocal}([x])^{\otimes p}\otimes \hat{\Ocal}^*([x])^{\otimes q})_{\alpha_1\dots \alpha_p,\beta_1\dots\beta_q}\nonumber\\
&\qquad\qquad := \hat{\Ocal}_{\alpha_1}([x])\circ \ldots \circ \hat{\Ocal}_{\alpha_p}([x])\circ \Xi^{-1}\circ \hat{\Ocal}^*_{\beta_1}([x])\circ \ldots \circ \hat{\Ocal}^*_{\beta_q}([x])\circ \Xi\,,
\end{align*}
for all $\alpha_1,\dots,\alpha_p,\beta_1,\dots,\beta_q \in \{1,\dots,\dim\rho\}$.
The individual components provide linear operator-valued tempered distributions on $\Dcal$, i.e., the Hilbert space remains the Hilbert space of the original projective quantum field $\hat{\Ocal}$, and the tensor products only encompass the index structure corresponding to $\rho$.
The field operators $\hat{\Ocal}(p,q;[x])$ are equipped with the finite-dimensional $\overline{\PGL_5\rr}$ representation ${\rho^{\otimes p}\otimes (\rho^*)^{\otimes q}}$, where $\rho^* = (\rho^{-1})^T$ is the dual (contragredient) representation of $\rho$.
We write
\begin{equation*}
\hat{\Ocal}^{\otimes p}\otimes (\hat{\Ocal}^*)^{\otimes q} := (U,\rho^{\otimes p}\otimes (\rho^*)^{\otimes q},\{[\hat{\Ocal}([x])^{\otimes p}\otimes \hat{\Ocal}^*([x])^{\otimes q}]\,|\,[x]\in\RP^4\})\,.
\end{equation*}
This is a composite projective quantum field of $\hat{\Ocal}$.
\end{example}
\vs

The spin-statistics theorem has been proven for QFTs on Poincaré geometry~\cite{streater2000pct} under additional assumptions such as energy positivity.
Translated into the present framework, it connects quantum state statistics with the multipliers of the projective representation $U|_{\PO((1),(3,1))}$.
Only completely symmetric (anti-symmetric) quantum states appear in the physical Hilbert space for bosons (fermions), which come with integer spin (half-odd integer spin) Poincaré group representations.
Based on algebraic QFT, there are indications that the spin-statistics theorem can be extended to globally hyperbolic space-times~\cite{Verch:2001bv}, which include the homogeneous Lorentz geometries of type $[g]_*\Gg(4,1)$.
In this work we consider the spin-statistics relation for composite projective quantum fields.
\vs

\begin{defi}\label{DefSpinStatistics}
Let $\hat{\Ocal}'$ be a composite projective quantum field of $\hat{\Ocal}$ with decomposition \eqref{EqCompositeFieldOpDecomposition} of the field operator representatives.
It \emph{obeys spin-statistics}, if for all $[x]\in\RP^4,\gamma=1,\dots,\dim\rho'$ and
\begin{enumerate}[(i)]
\item $\hat{\Ocal}'$ bosonic:
\begin{align*}
\hat{\Ocal}'_\gamma([x]) = &\; \frac{1}{p! \, q!}\sum_{\pi\in S_p}\sum_{\pi'\in S_q}\sum_{\alpha_1,\dots,\alpha_p=1}^{\dim\rho} \sum_{\beta_1,\dots,\beta_q=1}^{\dim\rho} a_{\gamma,\alpha_1\dots\alpha_p,\beta_1\dots\beta_q} \nonumber\\
&\qquad\qquad \times (\hat{\Ocal}([x])^{\otimes p}\otimes \hat{\Ocal}^*([x])^{\otimes q})_{\pi(\alpha_1)\ldots \pi(\alpha_p),\pi'(\beta_1)\ldots \pi'(\beta_q)}\,,
\end{align*}
\item $\hat{\Ocal}'$ fermionic:
\begin{align*}
\hat{\Ocal}'_\gamma([x]) =&\; \frac{1}{p! \, q!}\sum_{\pi\in S_p}\sum_{\pi'\in S_q}\sum_{\alpha_1,\dots,\alpha_p=1}^{\dim\rho} \sum_{\beta_1,\dots,\beta_q=1}^{\dim\rho} \mathrm{sgn}(\pi) \mathrm{sgn}(\pi') \, a_{\gamma,\alpha_1\dots\alpha_p,\beta_1\dots\beta_q}  \nonumber\\
&\qquad\qquad \times (\hat{\Ocal}([x])^{\otimes p}\otimes \hat{\Ocal}^*([x])^{\otimes q})_{\pi(\alpha_1)\ldots \pi(\alpha_p),\pi'(\beta_1)\ldots \pi'(\beta_q)}\,,
\end{align*}
\end{enumerate}
where $S_p$ denotes the degree-$p$ symmetric group.
Else, it violates spin-statistics.
\end{defi}
\vs

While spin is a property of the restricted projective representation $U|_{\PO((1),(3,1))}$, the multipliers $\omega_\pm$ appear for both the full $\PGL_5\rr$ and $\PO((1),(3,1))$, and need to agree by \Cref{LemmaMultiplierRestrictionAgreement}.
\Cref{DefSpinStatistics} thus provides a consistent formulation of the spin-statistics relation for composite projective quantum fields without internal degrees of freedom.
It can be naturally extended to projective quantum fields acting on multi-particle Hilbert spaces, which is beyond the present work.
\Cref{ThmSpinStatIrredPoincareIrred} (provided later) shows that certain projective quantum fields obey spin-statistics in the sense of \Cref{DefSpinStatistics}.

\subsubsection{Irreducible projective quantum fields}\label{SecIrreducibleQuantumFields}
Projective quantum fields can be characterized according to irreducibility of the $\pgl_5\rr$ representation $\tilde{\rho}$ corresponding to the $\overline{\PGL_5\rr}$ representation $\rho$.
\vs

\begin{defi}\label{DefIrreducibleQuantumField}
A projective quantum field $\hat{\Ocal}$ is \emph{irreducible}, if $\tilde{\rho}$ is irreducible as a $\pgl_5\rr$ representation.
\end{defi}
\vs

We consider Lie algebra instead of Lie group representations in \Cref{DefIrreducibleQuantumField}, since the representation $U$ is projective, which will allow for the consistent description of for instance Dirac fermions via projective quantum fields.
On the Lie algebra level, \Cref{DefIrreducibleQuantumField} coincides with the standard definition of irreducibility for Hilbert space operators~\cite{tung1985group}.

By the construction of projective quantum fields, $\tilde{\rho}$ is finite-dimensional.
All finite-dimensional, complex, irreducible representations of $\pgl_5\rr\cong \slfrak_5\rr$ are given by Schur modules for a pair of Young diagrams.
They are constructed from the fundamental representation $\CP^4_{\pgl_5\rr}$, for which $\pgl_5\rr$ acts on $\CP^4$ via projective matrix multiplication. 
We describe the related construction of partly symmetrized, partly anti-symmetrized composite projective quantum fields, closely following the standard construction of Schur modules~\cite{tung1985group, Bekaert:2006py, fulton2013representation}.
We first do so for a general projective quantum field $\hat{\Ocal}=(U,\rho,\{[\hat{\Ocal}([x])]\,|\,[x]\in\RP^4\})$, subsequently specifying $\rho$ further.

For a given Young diagram $\lambda = \{\lambda_1,\dots,\lambda_r\}$, $\lambda_i\geq \lambda_j$ if $i<j$, we denote by $\#\lambda = \sum_{i=1}^r\lambda_i$ its number of boxes and by $|\lambda|=r$ its number of rows.
Let $(\lambda,\lambda')$ be a pair of Young diagrams, $\lambda = \{\lambda_1,\dots,\lambda_r\},\lambda' = \{\lambda_1',\dots,\lambda_{r'}'\}$ with $\#\lambda = p$, $\#\lambda' = q$.
For the Young diagram $\lambda$ equipped with the canonical numbering of the boxes,%
\footnote{The canonical numbering is first along the boxes corresponding to $\lambda_1$, then $\lambda_2$ and so forth.
Young tableaux other than the canonical one result in isomorphic $\pgl_5\rr$ representations~\cite{fulton2013representation}.} 
we define subgroups of the degree-$p$ symmetric group $S_p$:
\begin{align*}
P_{\lambda} =&\; \{\pi\in S_p\,|\, \pi\text{ preserves each row}\}\,,\\
Q_{\lambda} =&\; \{\pi\in S_p\,|\, \pi\text{ preserves each column}\}\,,
\end{align*}
analogously for $\lambda'$ and $S_q$.
Elements $\pi\in S_p$, $\pi'\in S_q$ act on the components of $(p,q)$-tensors of field operator representatives as
\begin{equation*}
(\hat{\Ocal}([x])^{\otimes p}\otimes \hat{\Ocal}^*([x])^{\otimes q})_{\alpha_1\ldots \alpha_p,\beta_1\ldots \beta_q}\mapsto (\hat{\Ocal}([x])^{\otimes p}\otimes \hat{\Ocal}^*([x])^{\otimes q})_{\pi(\alpha_1)\ldots \pi(\alpha_p),\pi'(\beta_1)\ldots \pi'(\beta_q)}\,.
\end{equation*}
We denote the corresponding action of $(\pi,\pi')$ on projective $(p,q)$-tensors of field operator representatives by $e_{\pi}\otimes e_{\pi'}^*$.
Then we set
\begin{equation*}
a_{(\lambda,\lambda')} := \sum_{(\pi,\pi')\in P_\lambda\times P_{\lambda'}}  e_{\pi}\otimes e_{\pi'}^*\,,\qquad b_{(\lambda,\lambda')} := \sum_{(\pi,\pi')\in Q_\lambda\times Q_{\lambda'}}  \sgn(\pi)\, \sgn(\pi')\cdot  e_{\pi}\otimes e_{\pi'}^*\,.
\end{equation*}
The Young symmetrizer is defined as $c_{(\lambda,\lambda')} = a_{(\lambda,\lambda')} \circ b_{(\lambda,\lambda')}$.
It corresponds to symmetrizing projective $(p,q)$-tensors of field operator representatives along rows of the Young diagrams $(\lambda,\lambda')$ and anti-symmetrizing them along columns of $(\lambda,\lambda')$.

The field operators $c_{(\lambda,\lambda')} ( [\hat{\Ocal}([x])^{\otimes p}\otimes \hat{\Ocal}^*([x])^{\otimes q}])$ come with the $\overline{\PGL_5\rr}$ representation $\rho_{(\lambda,\lambda')}$, which is the usual Schur module construction applied to $\rho$ for the pair $(\lambda,\lambda')$.
We define the projective Schur quantum field of $\hat{\Ocal}$ as
\begin{equation*}
c_{(\lambda,\lambda')}(\hat{\Ocal}) := (\Hcal'_{(\lambda,\lambda')},(U,\rho_{(\lambda,\lambda')},\{c_{(\lambda,\lambda')} ( [\hat{\Ocal}([x])^{\otimes p}\otimes \hat{\Ocal}^*([x])^{\otimes q}])\,|\,[x]\in\RP^4\})\,,
\end{equation*}
which is a composite projective quantum field of $\hat{\Ocal}$.
It is irreducible, if $\tilde{\rho}=\CP^4_{\pgl_5\rr}$, such that $\tilde{\rho}_{(\lambda,\lambda')}=\CP^4_{(\lambda,\lambda')}$ is the usual complex $\pgl_5\rr$ module for the pair $(\lambda,\lambda')$ and therefore irreducible.
In this case we write $\hat{\Ocal}^{\mathrm{Schur}}_{(\lambda,\lambda')}:=c_{(\lambda,\lambda')}(\hat{\Ocal})$.
\vs

\begin{example}\label{ExampleSchurModuleDiffHilbertSpaces}
Consider projective $(2,0)$-tensors of field operator representatives and $\tilde{\rho} = \CP^4_{\pgl_5\rr}$.
Consider the Young diagrams
\begin{equation*}
\ytableausetup{centertableaux}
\lambda =\, \begin{ytableau}
       1 & 2
\end{ytableau}\;,\qquad\lambda'=\emptyset\,.
\end{equation*}
Then
\begin{equation*}
[\hat{\Ocal}^{\mathrm{Schur}}_{(\lambda,\lambda')}([x])] = [(\hat{\Ocal}_{\alpha_1}([x])\hat{\Ocal}_{\alpha_2}([x]) + \hat{\Ocal}_{\alpha_2}([x])\hat{\Ocal}_{\alpha_1}([x]))_{\alpha_1,\alpha_2}]\,,
\end{equation*}
and $\tilde{\rho}_{(\lambda,\lambda')} = \Pp (\cc^5 \todot \cc^5)_{\pgl_5\rr}$.
If
\begin{equation*}
\ytableausetup{centertableaux}
\lambda =\, \begin{ytableau}
       1 \\
       2
\end{ytableau}\;,\qquad\lambda'=\emptyset\,,
\end{equation*}
then
\begin{equation*}
[\hat{\Ocal}^{\mathrm{Schur}}_{(\lambda,\lambda')}([x])] = [(\hat{\Ocal}_{\alpha_1}([x])\hat{\Ocal}_{\alpha_2}([x]) - \hat{\Ocal}_{\alpha_2}([x])\hat{\Ocal}_{\alpha_1}([x]))_{\alpha_1,\alpha_2}]\,,
\end{equation*}
which comes with $\tilde{\rho}_{(\lambda,\lambda')} =  \Pp (\cc^5 \twedge\cc^5)_{\pgl_5\rr}$.
For $[X]\in\pgl_5\rr$ with $X$ its unique trace zero representative, the representations $\Pp (\cc^5\todot \cc^5)_{\pgl_5\rr}$ and $\Pp (\cc^5\twedge \cc^5)_{\pgl_5\rr}$ are given by $\Pp(X\otimes X)$ acting on $\Pp(\cc^5\odot \cc^5)$ and $\Pp(\cc^5\wedge \cc^5)$ via projective matrix multiplication on both factors, respectively.
\end{example}
\vs

\begin{lem}\label{LemmaTensorIrreps}
A projective quantum field $\hat{\Ocal}=(U,\rho,\{[\hat{\Ocal}([x])]\,|\,[x]\in\RP^4\})$ is irreducible, if and only if $\tilde{\rho}=\CP^4_{(\lambda,\lambda')}$ for a pair of Young diagrams $(\lambda,\lambda')$.
If $|\lambda|,|\lambda'| =5$, then $\CP^4_{(\lambda,\lambda')}$ is the trivial representation of $\pgl_5\rr$, and if $|\lambda|\geq 6$ or $|\lambda'|\geq 6$, then $\CP^4_{(\lambda,\lambda')} \equiv 0$, contradicting the $\dim \rho\neq 0$ assumption for projective quantum fields.
\end{lem}

\begin{proof}
The statements follow with $\pgl_5\rr\cong \slfrak_5\rr$ from the standard classification of finite-dimensional, complex, irreducible representations of $\slfrak_5\rr$~\cite{tung1985group,fulton2013representation}.
\end{proof}

\Cref{LemmaTensorIrreps} implies that with regard to the finite-dimensional representation $\rho$, all irreducible projective quantum fields are of the form of projective Schur quantum fields.
If \Cref{LemmaTensorIrreps} applies, we write the projective quantum field $\hat{\Ocal}$ as $\hat{\Ocal}_{(\lambda,\lambda')}$.
Providing further examples, many composite projective quantum fields of $\hat{\Ocal}$ decompose into the irreducible projective Schur quantum fields of $\hat{\Ocal}$.
\vs

\begin{prop}
All composite projective quantum fields $\hat{\Ocal}' = (U,\rho',\{[\hat{\Ocal}'([x])]\,|\,[x]\in\RP^4\})$ of $\hat{\Ocal} = (U,\CP^4_{\PGL_5\rr},\{[\hat{\Ocal}([x])]\,|\,[x]\in\RP^4\})$ decompose into direct sums of projective Schur quantum fields of $\hat{\Ocal}$, i.e., there are multiplicities $n_{(\lambda,\lambda')}\in \nn$ and representatives $\hat{\Ocal}'([x])$ of $[\hat{\Ocal}'([x])]$, such that
\begin{equation*}
\rho' \cong \bigoplus_{\substack{\lambda,\lambda'\\ \#\lambda=p,\#\lambda'=q }} n_{(\lambda,\lambda')} \rho_{(\lambda,\lambda')}
\end{equation*}
and for all $[x]\in\RP^4$:
\begin{equation*}
 \hat{\Ocal}'([x]) = \sum_{\substack{\lambda,\lambda'\\ \#\lambda=p,\#\lambda'=q }} n_{(\lambda,\lambda')} \hat{\Ocal}_{(\lambda,\lambda')}^{\mathrm{Schur}}([x])\,,
\end{equation*}
where $(p,q)$ is as in the decomposition \eqref{EqCompositeFieldOpDecomposition}.
\end{prop}

\begin{proof}
The finite-dimensional complex representation $\tilde{\rho}'$ of $\pgl_5\rr\cong \slfrak_5\rr$ is reducible and decomposes into a direct sum of the irreducible $\pgl_5\rr$ representations $\CP^4_{(\lambda,\lambda')}$~\cite{tung1985group,fulton2013representation}.
The Lie group representation $\rho'$ decomposes analogously to $\tilde{\rho}'$, since $\rho=\CP^4_{\PGL_5\rr}$.
Hence, due to the generalized unitary transformation behavior~\eqref{EqGlobalCovariance} the composite field operator representatives $\hat{\Ocal}'$ decompose as well, in agreement with the decomposition of $\tilde{\rho}'$.
Since the $\hat{\Ocal}'([x])$ are composite field operators of $\hat{\Ocal}$, their decomposition is into projective Schur quantum fields of $\hat{\Ocal}$.
\end{proof}

Obeying spin-statistics manifests for irreducible projective quantum fields as follows.
\vs

\begin{prop}\label{PropCombinedIrrepsSpinStatisticsConsistent}
Let $\hat{\Ocal}_{(\lambda,\lambda')}$ be an irreducible projective quantum field, which obeys spin-statistics.
Then
\begin{enumerate}[(i)]
\item for $\hat{\Ocal}_{(\lambda,\lambda')}$ bosonic both $\lambda$ and $\lambda'$ consist of a single row, or
\item for $\hat{\Ocal}_{(\lambda,\lambda')}$ fermionic both $\lambda$ and $\lambda'$ consist of a single column.
\end{enumerate}
\end{prop}

\begin{proof}
The statement is a direct consequence of \Cref{LemmaTensorIrreps} together with \Cref{DefSpinStatistics}.
\end{proof}

\subsection{Poincaré-irreducibility}\label{SecPoincareIrreducibility}
The behavior of field operators under Poincaré transformations is of particular interest for QFTs formulated on Poincaré geometry.
Fundamental quantum fields on Poincaré geometry are typically constructed by demanding that they transform irreducibly under Poincaré transformations.
\vs

\begin{defi}\label{DefPoincareIrreducible}
A projective quantum field $\hat{\Ocal} = (U,\rho,\{[\hat{\Ocal}([x])]\,|\,[x]\in\RP^4\})$ is \emph{Poincaré-irreducible}, if $U|_{\PO((1),(3,1))}$ is irreducible as a projective unitary $\PO((1),(3,1))$ representation.
The tuple $(U,\rho,\{[\hat{\Ocal}([x])]\,|\,[x]\in\RP^4\})$ is \emph{translation-invariant with respect to $\rho$}, if for all $[x]\in\RP^4, t\in\rr^4$: 
\begin{equation*}
[\rho([(0,t)])\hat{\Ocal}([x])] = [\hat{\Ocal}([x])]\,,
\end{equation*}
where $[(0,t)]\in\PGL_5\rr$ is defined by \Cref{EqPoincareTrafo55}.
\end{defi}
\vs

Poincaré-irreducibility is defined with respect to $U$, since translations on $\Xx((1),(3,1))$ can act non-trivially on the space-time arguments of field operators, even if $\hat{\Ocal}$ is translation-invariant with respect to $\rho$.

We can classify projective quantum fields, which are both irreducible in the sense of  \Cref{DefIrreducibleQuantumField} ($\tilde{\rho}$ irreducible as a $\pgl_5\rr$ representation) and Poincaré-irreducible.
For this we construct certain restricted projective quantum fields.
Given a projective quantum field $\hat{\Ocal}$, we set
\begin{equation*}
\hat{\Psi} = (U,\rho,\{[\hat{\Psi}([x])]\,|\,[x]\in\RP^4\})\,,
\end{equation*}
where the components of the representatives $\hat{\Psi}([x])$ are given by $\hat{\Psi}_\alpha([x]) = \hat{\Ocal}_\alpha([x])$ for $\alpha \notin J$ or $\hat{\Psi}_\alpha([x]) = 0$ for $\alpha\in J$. 
The index set $J\subset \{1,\dots,\dim\rho\}$ is defined, such that $\hat{\Psi}$ is translation-invariant with respect to $\rho$ and $J$ has minimal cardinality.
The field operators $\hat{\Psi}([x])$ are well-defined this way, based on \Cref{LemmaPoincareirredColumnOnly} as provided later.
We define the restriction of $\hat{\Psi}$ to a geometry such as $\Gg((1),(3,1))$ as in \Cref{DefFieldOperator}.
If $\hat{\Ocal} = \hat{\Ocal}_{(\lambda,\lambda')}$ is irreducible, we write $\hat{\Psi} = \hat{\Psi}_{(\lambda,\lambda')}$.
\vs

\begin{example}\label{ExampleCompositePsis}
Consider the irreducible projective quantum field $\hat{\Ocal}_{(\square,\emptyset)}$.
Its field operators are of the form $[\hat{\Ocal}_{(\square,\emptyset)}([x])] = [\hat{\Ocal}_{(\square,\emptyset),1}([x]),\dots,\hat{\Ocal}_{(\square,\emptyset),5}([x])]$.
We have
\begin{equation*}
[\hat{\Psi}_{(\square,\emptyset)}([x])] = [0,\hat{\Ocal}_{(\square,\emptyset),2}([x]),\dots,\hat{\Ocal}_{(\square,\emptyset),5}([x])]\,.
\end{equation*}
Analogously to the construction of the field operators $[\hat{\Ocal}_{(\lambda,\lambda')}^{\mathrm{Schur}}([x])]$ of a projective Schur quantum field from $[\hat{\Ocal}_{(\square,\emptyset)}([x])]$ as detailed in \Cref{SecIrreducibleQuantumFields}, we define $[\hat{\Psi}_{(\lambda,\lambda')}^{\mathrm{Schur}}([x])]$ from $[\hat{\Psi}_{(\square,\emptyset)}([x])]$.
We set
\begin{equation*}
\hat{\Psi}_{(\lambda,\lambda')}^{\mathrm{Schur}} := (U,\rho_{(\lambda,\lambda')},\{[\hat{\Psi}_{(\lambda,\lambda')}^{\mathrm{Schur}}([x])]\,|\, [x]\in\RP^4\})\,,
\end{equation*}
where $\tilde{\rho}_{(\lambda,\lambda')} = \CP^4_{(\lambda,\lambda')}$.
Upon restriction to Poincaré geometry $\Gg((1),(3,1))$, the set
\begin{equation*}
\{[\hat{\Psi}_{(\lambda,\lambda')}^{\mathrm{Schur}}([x])]\,|\,[x]\in\Xx((1),(3,1))\}
\end{equation*}
fulfils the unitary transformation property \eqref{EqLorentzTrafoVectorFieldBehavior}, such that $\hat{\Psi}_{(\lambda,\lambda')}^{\mathrm{Schur}}|_{\Gg((1),(3,1))}$ defines a restricted projective quantum field.
\end{example}
\vs

We can characterize irreducible, Poincaré-irreducible projective quantum fields as follows, which provides one of the main results of this work.
\vs

\begin{thm}\label{ThmSpinStatIrredPoincareIrred} 
Fermionic, irreducible, Poincaré-irreducible projective quantum fields behave under Poincaré transformations as Dirac fermions and obey spin-statistics.
If they are bosonic, they behave under Poincaré transformations as scalar or vector bosons, but violate spin-statistics as composite projective quantum fields.
\end{thm}
\vs

The proof of this theorem makes use of two more technical results, which we state and prove first.
For column-only $\lambda,\lambda'$, the field operators $[\hat{\Psi}_{(\lambda,\lambda')}([x])]$ transform under Poincaré transformations as specified by the following proposition.
\vs

\begin{prop}\label{PropPoincareTrafoBehaviorPsi}
Assume both $\lambda,\lambda'$ consist of a single column.
Then the Poincaré transformation $[(\Lambda,t)]\in \PO((1),(3,1))$ acts on the field operators $[\hat{\Psi}_{(\lambda,\lambda')}([x])]$ for $[x]\in\Xx((1),(3,1))$ as follows:
\begin{enumerate}[(i)]
\item for $\hat{\Psi}_{(\lambda,\emptyset)}$ and $\hat{\Psi}_{(\emptyset,\lambda)}$ with $\#\lambda\in \{0,4\}$ via the $(0,0)$ (scalar) representation, with $\#\lambda\in \{1,3\}$ via the $(1/2,0)\oplus (0,1/2)$ (Dirac fermion) representation, and with $\#\lambda=2$ via the $(1/2,1/2)$ (vector) representation  (the half-integer pairs indicating the spin of the Poincaré group representations),
\item for $\hat{\Psi}_{(\lambda,\lambda')} = \hat{\Psi}_{(\lambda,\emptyset)}\otimes \hat{\Psi}_{(\emptyset,\lambda')}$ via the tensor product of the $\PO((1),(3,1))$ representations of the two factors.
\end{enumerate}
If any $\#\lambda,\#\lambda'\geq 5$, the corresponding field operators are all zero.
\end{prop}

\begin{proof}
Translations $[(0,t)]$ act by construction only on the space-time arguments of the $[\hat{\Psi}_{(\lambda,\lambda')}([x])]$.
Lorentz transformations $[(\Lambda,0)]$ can act non-trivially via $\rho_{(\lambda,\lambda')}$.
On Lie algebra level, the action of the complexification $\ofrak(3,1)_\cc\cong \slfrak_2\cc\oplus \overline{\slfrak_2\cc}$ is to be considered.
We note that $\tilde{\Extalt}^n(0\oplus \cc^2\oplus \overline{\cc}^2)_{0\oplus \slfrak_2\cc \oplus \slfrak_2\cc}\cong \tilde{\Extalt}^n(\cc^2\oplus \overline{\cc}^2)_{\slfrak_2\cc\oplus \slfrak_2\cc}$, where $\overline{\cc}^2_{\slfrak_2\cc}$ is the complex conjugate of the fundamental representation $\cc^2_{\slfrak_2\cc}$.
Exterior powers of direct sums of Lie algebra representations decompose by basic module theory as
\begin{equation*}
\tilde{\Extalt}^n (V\oplus W)\cong \bigoplus_{p=0}^n \tilde{\Extalt}^p V\totimes \tilde{\Extalt}^{n-p} W\,,
\end{equation*}
such that
\begin{equation*}
\tilde{\Extalt}^n \big(\cc^2_{\slfrak_2\cc}\oplus \overline{\cc}^2_{\slfrak_2\cc}\big) \cong \bigoplus_{p=0}^n \tilde{\Extalt}^p \cc^2_{\slfrak_2\cc}\totimes \tilde{\Extalt}^{n-p} \overline{\cc}^2_{\slfrak_2\cc}\,.
\end{equation*}
We have that $\Extalt^0\cc^2_{\slfrak_2\cc} \cong \Extalt^2\cc^2_{\slfrak_2\cc} \cong \cc$ is the trivial representation and $\Extalt^p\cc^2\cong 0$ for $p\geq 3$ is trivial as a vector space, which yields the claim.
\end{proof}

\begin{lem}\label{LemmaPoincareirredColumnOnly}
Let $\hat{\Ocal}_{(\lambda,\lambda')}=(U,\rho_{(\lambda,\lambda')},\{[\hat{\Ocal}_{(\lambda,\lambda')}([x])]\,|\,[x]\in\RP^4\})$ be an irreducible projective quantum field, which is Poincaré-irreducible.
Then,
\begin{equation*}
\hat{\Ocal}_{(\lambda,\lambda')}|_{\Gg((1),(3,1))} = \hat{\Psi}_{(\lambda,\lambda')}|_{\Gg((1),(3,1))}
\end{equation*}
for a pair $(\lambda,\lambda')$ of column-only Young diagrams, where one of the Young diagrams $\lambda,\lambda'$ must be empty.
$\hat{\Ocal}_{(\lambda,\lambda')}$ is fermionic (bosonic), if and only if $\#\lambda+\#\lambda'$ is uneven (even).
\end{lem}

\begin{proof}
Let $p:=\#\lambda$ and $q:=\#\lambda'$, and set
\begin{equation*}
[(0,\tilde{t})]:=\pfrak\left(\begin{matrix}
0 & \\
\tilde{t} & 0_{4\times 4}
\end{matrix}\right)\in \pofrak((1),(3,1))
\end{equation*}
which describes a translation in $\Xx((1),(3,1))$ on Lie algebra level, $\tilde{t}\in\rr^4$.
Complexified translations act analogously.
$\hat{\Ocal}_{(\lambda,\lambda')}$ needs to be translation-invariant with respect to $\rho_{(\lambda,\lambda')}$ by Poincaré-irreducibility, since the projective unitary irreducible representation of the Poincaré group $\PO((1),(3,1))$ are induced from the Lorentz group $\Oo(3,1)$ as the corresponding little group.
By the construction of $\rho_{(\lambda,\lambda')}$, the components $\hat{\Ocal}_{(\lambda,\lambda'),\alpha}([x])$ of the field operator representatives can be equivalently indexed by $\alpha_1,\dots,\alpha_p,\beta_1,\dots,\beta_q=1,\dots,5$ for some $p,q\in\nn$ instead of $\alpha$ (analogously to $(p,q)$-tensors), which obey partial symmetry, partial anti-symmetry upon permutations, as dictated by $\lambda,\lambda'$.
We choose the index ordering, such that $\pgl_5\rr$-elements act via matrix multiplication with respect to each of the indices $\alpha_1,\dots,\alpha_p,\beta_1,\dots,\beta_q$.
The demanded triviality of the translation action requires $\hat{\Ocal}_{(\lambda,\lambda'),\alpha_1\dots\alpha_p,\beta_1\dots\beta_q}([x])=0$ for all $[x]\in\Xx((1),(3,1))$, whenever at least one of the indices $\alpha_1,\dots,\alpha_p,\beta_1,\dots,\beta_q$ is 1.
Therefore, $\hat{\Ocal}|_{\Gg((1),(3,1))}=\hat{\Psi}_{(\lambda,\lambda')}|_{\Gg((1),(3,1))}$, such that $\hat{\Psi}$ is translation-invariant with respect to $\rho_{(\lambda,\lambda')}$.

We next verify Poincaré-irreducibility with respect to Lorentz transformations.
We set
\begin{equation*}
[(\tilde{\Lambda},0)]:=\pfrak\left(\begin{matrix}
0 &\\
& \tilde{\Lambda}
\end{matrix}\right)\in \pofrak((1),(3,1))
\end{equation*}
for $\tilde{\Lambda}\in\ofrak(3,1)$ and note that $\tilde{\rho}_{(\square,\emptyset)}([(\tilde{\Lambda},0)])$ acts on elements of the space 
\begin{equation*}
C:=\{[0,c_1,\dots,c_4]\neq [0]\,|\, c_\mu\in\cc\}\subset \CP^4
\end{equation*}
as for a Dirac fermion via the massive spin $(1/2,0)\oplus (0,1/2)$ Poincaré representation, see \Cref{PropPoincareTrafoBehaviorPsi}.
This describes also the action of $\tilde{\rho}_{(\square,\emptyset)}([(\tilde{\Lambda},0)])$ on $[\hat{\Psi}_{(\square,\emptyset)}([x])]$ and takes the equivalence of $\tilde{\rho}$ as a complex $\pgl_5\rr$ representation with the corresponding $\pgl_5\cc$ representation into account. 

The general representations $\tilde{\rho}_{(\lambda,\lambda')}$ are constructed via Young symmetrization and dualization from $\tilde{\rho}_{(\square,\emptyset)}$.
The Lie subalgebra $\pfrak(0\oplus \gl_4\rr)$ acts irreducibly on the Schur module $C_{(\lambda,\lambda')}$ of $C$ via $\tilde{\rho}_{(\lambda,\lambda')}$.
Restricting to the Lorentz subalgebra $\pfrak(0\oplus \ofrak(3,1))$, the $\tilde{\rho}_{(\lambda,\lambda')}|_{\pfrak(0\oplus \ofrak(3,1))}$ act not necessarily irreducibly on $C_{(\lambda,\lambda')}$.
Let $\lambda/\mu$ be the sum of those Young diagrams $\nu$ for which $\nu\cdot\mu$ (the product of Young diagrams describing the corresponding tensor product of representations) contains multiples of $\lambda$ upon its decomposition into Schur modules.
Define the formal sum
\begin{equation*}
\ytableausetup{centertableaux,smalltableaux}
\Delta =1 + \, \begin{ytableau}
       {} & {}
\end{ytableau} + \begin{ytableau}
       {} & {} & {} & {}
\end{ytableau} + \begin{ytableau}
       {} & {}\\
       {} & {}
\end{ytableau}+ \dots \,.
\end{equation*}
The branching rule for general $\GL_n\rr\downarrow \Oo(n)$~\cite{Bekaert:2006py, King:1975vf} can be applied, since $\glfrak_n\rr\cong \pfrak(0\oplus\glfrak_n\rr)$ and $\ofrak(n)\cong \pfrak(0\oplus \ofrak(n))$ as Lie algebras.
This yields the following decomposition of the restriction of $\tilde{\rho}_{(\lambda,\lambda')}$ to $\pfrak(0\oplus \ofrak(3,1))$ into irreducible $\pfrak(0\oplus \ofrak(3,1))$ representations:
\begin{align}
\tilde{\rho}_{(\lambda,\lambda')}|_{\pfrak(0\oplus \ofrak(3,1))}\cong&\; \tilde{\rho}_{(\lambda/\Delta,\lambda'/\Delta)}^{\pfrak(0\oplus\ofrak(3,1))}\nonumber\\
= &\; \tilde{\rho}_{(\lambda/\Delta,\emptyset)}^{\pfrak(0\oplus\ofrak(3,1))}\totimes \tilde{\rho}_{(\emptyset,\lambda'/\Delta)}^{\pfrak(0\oplus\ofrak(3,1))}\nonumber\\
= &\; \big(\tilde{\rho}_{(\lambda,\emptyset)}^{\pfrak(0\oplus\ofrak(3,1))} \oplus \tilde{\rho}_{(\lambda/\{2\},\emptyset)}^{\pfrak(0\oplus\ofrak(3,1))}\oplus \tilde{\rho}_{(\lambda/\{4\},\emptyset)}^{\pfrak(0\oplus\ofrak(3,1))}\oplus \ldots\big) \nonumber\\
&\qquad\qquad \totimes \big(\tilde{\rho}_{(\emptyset,\lambda')}^{\pfrak(0\oplus\ofrak(3,1))} \oplus \tilde{\rho}_{(\emptyset,\lambda'/\{2\})}^{\pfrak(0\oplus\ofrak(3,1))}\oplus \tilde{\rho}_{(\emptyset,\lambda'/\{4\})}^{\pfrak(0\oplus\ofrak(3,1))}\oplus \ldots\big)\,,\label{EqBranchingO31}
\end{align}
where the representations on the right-hand side denote Young-symmetrized, complex tensor product representations of $\pfrak(0\oplus \ofrak(3,1)) \curvearrowright C$, which are equivalent to those constructed from $\pfrak(0\oplus \ofrak(3,1))_\cc\curvearrowright C$.
That is, they are given by $\pfrak(0\oplus \ofrak(3,1))_\cc$ acting via matrix multiplication on each of the tensor product factors in
\begin{align*}
&C_{(\lambda/\Delta,\lambda'/\Delta)}= (C_{(\lambda,\emptyset)}\oplus C_{(\lambda/\{2\},\emptyset)}\oplus C_{(\lambda/\{4\},\emptyset)}\oplus \ldots)\nonumber\\
&\qquad\qquad\qquad\qquad\qquad\qquad \otimes (C_{(\emptyset,\lambda')}\oplus C_{(\emptyset,\lambda'/\{2\})}\oplus C_{(\emptyset,\lambda'/\{4\})}\oplus \ldots)\,.
\end{align*}
The representation $\tilde{\rho}_{(\lambda,\lambda')}|_{\pfrak(0\oplus \ofrak(3,1))}$ comes from a Poincaré-irreducible projective quantum field, if and only if the decomposition \eqref{EqBranchingO31} on the right-hand side is into a single tensor product and not a sum of them.
This is the case, if and only if $\lambda$ and $\lambda'$ are empty or consist of columns only, such that $\lambda/\{2\},\lambda'/\{2\},\lambda/\{4\},\lambda'/\{4\},\dots = \emptyset$.
Moreover, if both $\lambda,\lambda'\neq \emptyset$, then $\tilde{\rho}_{(\lambda,\lambda')}|_{\pfrak(0\oplus \ofrak(3,1))}$ is again reducible, since in that case direct sums of the different spin representations appear in the direct sum decomposition of the tensor product of spin representations.
Therefore, one of $\lambda,\lambda'$ must be empty.

The last statement on $\#\lambda+\#\lambda'$ is a consequence of \Cref{PropPoincareTrafoBehaviorPsi} together with the classification of irreducible Poincaré group representations.
To see this, denote by $s(\lambda)$ the total spin of the representation $\tilde{\rho}_{(\lambda,\emptyset)}|_{\pfrak(0\oplus \ofrak(3,1))}$, i.e., $s(\lambda)=0$ for the $(0,0)$ representation, $s(\lambda)=1/2$ for the $(1/2,0)\oplus(0,1/2)$ representation and $s(\lambda)=1$ for the $(1/2,1/2)$ representation.
The Poincaré Lie algebra representation $\tilde{\rho}_{(\lambda,\lambda')}|_{\pofrak((1),(3,1))}$ has total spin $s(\lambda)+s(\lambda')$, as can be seen with the Littlewood-Richardson rule to decompose the tensor product 
\begin{equation*}
\tilde{\rho}_{(\lambda,\lambda')}|_{\pofrak((1),(3,1))} = \tilde{\rho}_{(\lambda,\emptyset)}|_{\pofrak((1),(3,1))}\totimes \tilde{\rho}_{(\emptyset,\lambda')}|_{\pofrak((1),(3,1))}\,.
\end{equation*}
By the classification of the irreducible projective unitary representations of the Poincaré group~\cite{Weinberg:1995mt}, a fermionic multiplier of $U$ is consistent with the spin of $\tilde{\rho}_{(\lambda,\lambda')}|_{\pofrak((1),(3,1))}$, if and only if $s(\lambda)+s(\lambda')$ is a half-integer.
Explicit computation with \Cref{PropPoincareTrafoBehaviorPsi} yields that this is the case, if and only if $\#\lambda+\#\lambda'$ is uneven.
The same argument leads to the claim for a bosonic multiplier of $U$ for $\#\lambda+\#\lambda'$ even.
\end{proof}

The proof of \Cref{ThmSpinStatIrredPoincareIrred} is now straight-forward.

\begin{proof}[Proof of \Cref{ThmSpinStatIrredPoincareIrred}]
By \Cref{LemmaPoincareirredColumnOnly}, irreducible, Poincaré-irreducible projective quantum fields are upon restriction to Poincaré geometry by of the form $\hat{\Psi}_{(\lambda,\lambda')}|_{\Gg((1),(3,1))}$ for column-only Young diagrams $\lambda,\lambda'$ with the additional constraints that one of $\lambda,\lambda'$ must be empty.
In the fermionic case $\#\lambda+\#\lambda'$ must be uneven, in the bosonic case even.
The application of Propositions \ref{PropCombinedIrrepsSpinStatisticsConsistent} and \ref{PropPoincareTrafoBehaviorPsi} then yields the claim.
\end{proof}

A corollary of \Cref{ThmSpinStatIrredPoincareIrred} finally characterizes $\hat{\Psi}|_{\Gg((1),(3,1))}$ as restricted projective quantum fields.
\vs

\begin{cor}
Constructed from an arbitrary projective quantum field $\hat{\Ocal}$, $\hat{\Psi}|_{\Gg((1),(3,1))}$ is a restricted projective quantum field.
\end{cor}

\begin{proof}
Any finite-dimensional, complex $\overline{\PGL_5\rr}$ representation $\rho$ is reducible, factoring through a finite-dimensional representation of $\PGL_5\rr$.
The proof of \Cref{LemmaPoincareirredColumnOnly} shows the claim for the irreducible representations $\rho_{(\lambda,\lambda')}$.
Therefore, it holds also for $\hat{\Psi}$ constructed from arbitrary projective quantum fields $\hat{\Ocal}$.
\end{proof}


\section{Further questions}\label{SecConclusions}
This work provided an axiomatic formulation of projective quantum fields on subgeometries of four-dimensional real projective geometry.
It has been based upon their well behavior under geometry deformations and limits.
The setting also allowed us to show that all projective correlators and their expectation values remain well-defined under such geometry transformations, even if their support can shift to space-time boundaries and other lower-dimensional space-time subspaces.
We explored a range of structural properties of projective quantum fields and related the framework to more traditional formulations of quantum fields on e.g.~Poincaré geometry.
\vs

The results of this work suggest a few further questions:
\begin{itemize}
\item Based on the absence of coordinate singularities, the description of space-time geometries as subgeometries of four-dimensional real projective geometry appears mathematically beneficial compared to the more common description of homogeneous space-time geometries and their deformations and limits via contractions. 
Arguing in favor of conjugacy limits, they appear naturally from canonical constructions in the framework of geometries, which is in contrast to contractions.
From the viewpoint of representation theory, the related Chabauty topology provides a natural setting to describe limits of subgroups and their representations.
Can one provide further, more physically motivated arguments in favor of conjugacy limits?
\item The well behavior of projective quantum fields under geometry limits also motivates the potential usability of projective quantum fields in holographic correspondences such as the AdS/CFT correspondence, which rest upon suitable maps between bulk and boundary fields.
It appears worthwhile to investigate in how far geometry limits of projective quantum fields can at least partly provide such maps as well as new insights into e.g.~the Poincaré limit of holographic correspondences.
In this regard, we note that projective geometry techniques have already been shown to facilitate the computation of boundary fields, see e.g.~\cite{Bekaert:2012vt, Bekaert:2013zya}.
\item Can the consideration of the asymptotic behavior of more common QFTs such as quantum electrodynamics yield insights into the physicality of the projective geometry setting explored in this work?
\item The definition of projective quantum fields did not implement global hyperbolicity of space-times, causality preservation or a variant of the spectrum condition.
How can these properties, which provide essential ingredients of the mathematical formulation of QFTs, be consistently incorporated into the general framework of this work?
\item The presented framework rests upon the homogeneity of space-time geometries.
Yet, continuity under geometry deformations implies that small geometry deformations alter projective correlators only little.
We thus expect our results to hold approximately also for inhomogeneous space-times, if deformations are small.
Proximity can be defined via the structure groups with respect to available metrics on the space of closed subgroups of $\PGL_5\rr$~\cite{biringer2018metrizing}.
What is a more general, suitable description of projective quantum fields on inhomogeneous, curved space-time geometries generalizing the one given in this work?
\item For composite projective quantum fields, \Cref{ThmSpinStatIrredPoincareIrred} indicates that the spin-statistics relation can at least partly be understood based on representation theory.
How far can a generalization of this go?
\end{itemize}
\vs

Addressing these questions could shed further light onto non-trivial physical implications of the mathematical constructions put forward in this work.

\section*{Acknowledgments}
I acknowledge fruitful exchange on the present work with S.~Flörchinger, L.~Hahn, D.~Roggenkamp, M.~Salmhofer, S.~Schmidt and A.~Wienhard.
This work is funded by the Deutsche Forschungsgemeinschaft (DFG, German Research Foundation) under Germany’s Excellence Strategy EXC 2181/1–390900948 (the Heidelberg STRUCTURES Excellence Cluster) and the Collaborative Research Centre, Project-ID No.~273811115, SFB 1225 ISOQUANT.

\section*{Appendices}

\appendix
\begin{appendices}


\section{Proofs for the relation between conjugacy limits and contractions}\label{AppendixContractionsProofs}
This appendix provides the proofs of \Cref{LemmaNoContractionO3toO111} and \Cref{ThmMultipleContrations}.
We often omit the explicit notation of projective equivalence classes.

\subsection{Proof of \Cref{LemmaNoContractionO3toO111}}\label{AppendixPfLemmaNoContractionO3toO111}
To prove \Cref{LemmaNoContractionO3toO111} we first show two propositions.
\vs

\begin{prop}\label{PropMaxLieSubalgebraContractionInv}
The maximal Lie subalgebra of a Lie algebra $\hfrak$ with commutator $[\cdot,\cdot]$, which is invariant under contraction along a subalgebra $\tfrak$ of $\hfrak$, is isomorphic to a sum of $\tfrak$ and an Abelian subalgebra $\sfrak$ of $\hfrak$ with $[\tfrak,\sfrak]\subset \sfrak$.
The sum is not necessarily direct.
\end{prop}

\begin{proof}
Let $\hfrak'$ with commutator $[\cdot,\cdot]'$ denote the contraction of $\hfrak$ along $\tfrak$, and $\tfrak^c$ be the vector space complement of $\tfrak$: $\vect(\hfrak) = \vect(\tfrak) \oplus \vect(\tfrak^c)$.
Clearly, $\tfrak$ is a subalgebra of $\hfrak$ and $\hfrak'$, which remains invariant under the contraction along itself.

Assume there exists a vector subspace $\sfrak\subset \tfrak^c$, such that the sum $\tfrak + \sfrak = (\vect(\tfrak)\oplus \vect(\sfrak),[\cdot,\cdot])$ is invariant under the contraction.
By the definition of a contraction, we have for $X\in \tfrak, Y\in \sfrak$: $[X,Y]'=[X,Y]$ if and only if $[X,Y]\in \tfrak^c$ or $[X,Y]=0$. 
Thus, for $\tfrak + \sfrak$ to be a subalgebra we need $[X,Y]\in \sfrak$ for $X\in\tfrak, Y\in \sfrak$, since $(\tfrak + \sfrak)\cap \tfrak^c=\sfrak$.
For this we also require that for all $X,Y\in \sfrak$: $[X,Y]'=[X,Y]$ by the demanded invariance under the contraction.
By the contraction definition, we find for such $X,Y$: $[X,Y]'=0$, i.e., $[X,Y]=0$, such that $\sfrak$ must be an Abelian subalgebra.
\end{proof}

Let $H$ denote a Lie group with Lie algebra $\hfrak$ and restrict to the indefinite orthogonal groups or deformations or conjugacy limits of such.
For the $KBH$ decomposition of $\PGL_m\rr$ this implies that $B$ is the subgroup of diagonal projective matrices, $K$ a maximal compact subgroup~\cite{cooper2018limits}.
Let $b_n = \exp( n X_b)\in B$ and $\lfrak$ be the conjugacy limit of $\hfrak$ via $b_n$ for $n\to\infty$.
We have the limit decomposition $\lfrak = \zfrak\oplus \nfrak_+$~\cite{cooper2018limits} with subalgebras
\begin{equation*}
\zfrak = \{X\in\hfrak\,|\, [X_b,X]=0\}\,,\quad \nfrak_+=\{X\in\pgl_m\rr\,|\, \lim_{n\to\infty} \exp(- n X_b) X \exp( n X_b)=0\}\,.
\end{equation*}
By definition, $\zfrak$ is the maximal subalgebra of $\hfrak$, which is invariant under the conjugacy limit via $b_n$.
\vs

\begin{prop}\label{PropSubalgebraContraction}
We use the notation from the proof of \Cref{PropMaxLieSubalgebraContractionInv}. 
Assume the contracted algebra $\hfrak'$ and the conjugacy limit algebra $\lfrak$ are isomorphic.
If there exists no non-trivial Lie algebra morphism from the radical of $\hfrak'$ into $\hfrak$, then the subalgebra $\tfrak$ of the contraction $\hfrak\to\hfrak'$ must be isomorphic to $\zfrak$.
\end{prop}

\begin{proof}
By \Cref{PropMaxLieSubalgebraContractionInv}, the maximal subalgebra of $\hfrak$, which is invariant under the contraction along $\tfrak$, is isomorphic to $\tfrak +\sfrak$ for an Abelian subalgebra $\sfrak\subset \tfrak^c$ with $[\tfrak,\sfrak] \subset \sfrak$.
We have $[\tfrak^c,\sfrak]'=0$ due to the contraction definition, such that $[\hfrak',\sfrak]'\subset \sfrak$ and $\sfrak$ is an Abelian ideal in $\hfrak'$.
We can Levi-decompose $\hfrak'=\mathrm{rad}(\hfrak')\oplus \hfrak'_{\mathrm{ss}}$, where $\mathrm{rad}(\hfrak')$ denotes the radical of $\hfrak'$ and $\hfrak'_{\mathrm{ss}}$ is semi-simple.
Thus, $\sfrak\subset \mathrm{rad}(\hfrak')$.
If there exists no non-trivial Lie algebra morphism $\mathrm{rad}(\hfrak')\to \hfrak$, we have $\sfrak = \{0\}$.

For $\hfrak'$ and $\lfrak$ to be isomorphic, the corresponding maximal invariant subalgebras of $\hfrak$ must be isomorphic (invariance once under the contraction, once under the conjugacy limit), such that $\tfrak\cong \zfrak$.
\end{proof}

\begin{proof}[Proof of \Cref{LemmaNoContractionO3toO111}]
The invariant subalgebra of $\pofrak(m)$ for the conjugacy limit via $b_n=\Pp\,\mathrm{diag}(1,\exp(n),\dots,\exp((m-1)n))$ is trivial: $\zfrak = \{0\}$.
We assume there exists a Lie algebra $\hfrak'\cong \pofrak((1),\dots,(1))$ which arises as a contraction of $\pofrak(m)$ along some subalgebra $\tfrak\subset \pofrak(m)$.
Due to $\hfrak' = \mathrm{rad}(\hfrak')$, there is an isomorphism $\eta:\mathrm{rad}(\hfrak') \to \pofrak((1),\dots,(1))$.
As vector spaces, $\pofrak((1),\dots,(1))$ and $\pofrak(m)$ are isomorphic, but not as Lie algebras for $m\geq 3$: the former is nilpotent, the latter simple.
There exists no non-trivial Lie algebra morphism $\pofrak((1),\dots,(1))\to \pofrak(m)$, thus also none of the form $\hfrak'\to \pofrak(m)$.
The subalgebra $\tfrak$ must be isomorphic to $\zfrak$ by \Cref{PropSubalgebraContraction}, hence trivial.
Therefore, for all $X,Y\in \hfrak'$ the contracted commutators are trivial: $[X,Y]' = 0$.
The Lie algebra $\hfrak'$ is Abelian, while $\pofrak((1),\dots,(1))$ is non-Abelian for $m\geq 3$, which is a contradiction.
No isomorphism between them exists.
\end{proof}

\subsection{Proof of \Cref{ThmMultipleContrations}}\label{AppendixPfThmMultipleContrations}

We now turn to the proof of \Cref{ThmMultipleContrations}.

\begin{proof}[Proof of \Cref{ThmMultipleContrations}]
Instead of $\pofrak(p,q)$ we first consider $\pofrak(m)$ for $m=p+q$, which is formed by the fixed points of the Cartan involution, i.e., those $X\in\pgl_m\rr$ with $X = \theta(X)  = -  X^T$.
Let $b_n\in B$ be of the form $b_{n,ii}/b_{n,(i+1)(i+1)} = 1$ for all $n$ if $i\notin S$ and $b_{n,ii}/b_{n,(i+1)(i+1)}\to 0$ for $n\to\infty$ if $i\in S$, where $S\subset \{1,\dots,m-1\}$.
The claim for $\pofrak(p,q)$ and general sequences $b_n\in \PGL_m\rr$ follows as detailed at the end of the proof.
We set $s=\# S$ and denote the $l$-th smallest integer in $S$ by $i_l$.
Then there exist matrices $b_n^{(l)}\in B$, such that
\begin{equation*}
b_n = b_n^{(s)}\cdot b_n^{(s-1)}\cdot \ldots \cdot b_n^{(1)}
\end{equation*}
with $b_{n,ii}^{(l)}/b_{n,(i+1)(i+1)}^{(l)} = 1$ for all $n$ if $i\neq i_l$ and $b_{n,ii}^{(l)}/b_{n,(i+1)(i+1)}^{(l)}\to 0$ for $n\to \infty$ if $i=i_l$.
We denote the conjugacy limit of the Lie algebra $\pofrak(m)\subset \pgl_m\rr$ via the sequence $b_n^{(l)}\cdot b_n^{(l-1)}\cdot \ldots \cdot b_n^{(1)}$ for $n\to \infty$ by $\ofrak_l$ and note that $\ofrak_s = \ofrak'$ is the conjugacy limit of $\pofrak(m)$ via $b_n$.

We show by induction in $l$ that each $\ofrak_l$ is isomorphic to a composition of $l$ contractions of $\pofrak(m)$.
We begin with $l=1$ and note that $\ofrak_1 = \zfrak_1\oplus \nfrak_{+,1}$~\cite{cooper2018limits} for the invariant subalgebra $\zfrak_1 = \pfrak(\ofrak(i_1)\oplus \ofrak(m-i_1))$ and
\begin{equation*}
\nfrak_{+,1} = \pfrak\left(\begin{matrix}
0 & \\
\rr^{(m-i_1)\times i_1} & 0
\end{matrix}\right)\,,
\end{equation*}
which is an Abelian ideal in $\ofrak_1$, since commutators of $\zfrak_1$- with $\nfrak_{+,1}$-elements read
\begin{equation*}
\Pp \left[\left(\begin{matrix}
X & 0\\
0 & X'
\end{matrix}\right), \left(\begin{matrix}
0 & 0\\
Y & 0
\end{matrix}\right)\right] = \Pp \left(\begin{matrix}
0 & 0\\
X'Y - YX & 0
\end{matrix}\right)
\end{equation*}
and are again in $\nfrak_{+,1}$.
We contract $\pofrak(m)$ along $\zfrak_1$, which yields the contracted Lie algebra $\hfrak_1$ with commutator $[\cdot, \cdot]_1$.
Explicitly, we find for $X_1,X_2\in \ofrak(i_1)$, $X_1',X_2'\in \ofrak(m-i_1)$ and $Y_1,Y_2 \in \rr^{(m-i_1)\times i_1}$:
\begin{subequations}
\begin{align}
\Pp\left[\left(\begin{matrix}
X_1 & 0\\
0 & X_1'
\end{matrix}\right),
\left(\begin{matrix}
X_2 & 0\\
0 & X_2'
\end{matrix}\right)\right]_1 = &\; \Pp\left[\left(\begin{matrix}
X_1 & 0\\
0 & X_1'
\end{matrix}\right),
\left(\begin{matrix}
X_2 & 0\\
0 & X_2'
\end{matrix}\right)\right]\nonumber\\
=&\;\Pp\left(\begin{matrix}
[X_1,X_2] & 0\\
0 & [X_1',X_2']
\end{matrix}\right)\,,\label{EqContractedCommA}\\
\Pp\left[\left(\begin{matrix}
X_1 & 0\\
0 & X_1'
\end{matrix}\right),
\left(\begin{matrix}
0 & - Y_1^T \\
Y_1 & 0
\end{matrix}\right)\right]_1 = &\; \Pp\left[\left(\begin{matrix}
X_1 & 0\\
0 & X_1'
\end{matrix}\right),
\left(\begin{matrix}
0 & - Y_1^T\\
Y_1 & 0
\end{matrix}\right)\right]\nonumber\\
=&\; \Pp\left(\begin{matrix}
0 & Y_1^T X_1' - X_1 Y_1^T  \\
X_1' Y_1 - Y_1 X_1 & 0
\end{matrix}\right)\,,\label{EqContractedCommB}\\
\Pp\left[\left(\begin{matrix}
0 & -Y_1^T\\
Y_1 & 0
\end{matrix}\right),
\left(\begin{matrix}
0 & -Y_2^T \\
Y_2 & 0
\end{matrix}\right)\right]_1 = &\; 0\,,\label{EqContractedCommC}
\end{align}
\end{subequations}
where $X_i = \theta(X_i), X_i' = \theta(X_i')$.
We define the map $\sigma_1:\hfrak_1\to \ofrak_1$,
\begin{equation*}
\sigma_1\left( \Pp\left(\begin{matrix}
X_1 & -Y_1^T\\
Y_1 & X_1'
\end{matrix}\right)\right) := \Pp\left(\begin{matrix}
X_1 & 0\\
Y_1 & X_1'
\end{matrix}\right)\,,
\end{equation*}
which is a vector space isomorphism with inverse
\begin{equation}\label{EqSigma1Inverse}
\sigma_1^{-1}\left( \Pp\left(\begin{matrix}
X_1 & 0\\
Y_1 & X_1'
\end{matrix}\right)\right) = \Pp\left(\begin{matrix}
X_1 & - Y_1^T\\
Y_1 & X_1'
\end{matrix}\right)\,.
\end{equation}
Comparing the contracted commutators \eqref{EqContractedCommA} to \eqref{EqContractedCommC} with the matrix commutators of $\ofrak_1$, we see that $\sigma_1$ is a Lie algebra morphism, so that $\ofrak_1$ and $\hfrak_1$ are isomorphic as Lie algebras.

For the induction step, suppose we have an isomorphism $\sigma_l:\hfrak_l\to \ofrak_l$. 
Using $\sigma_l$, we construct an isomorphism $\sigma_{l+1}:\hfrak_{l+1}\to \ofrak_{l+1}$.
Denote the commutator of $\hfrak_l$ by $[\cdot,\cdot]_l$ and set $j_l=i_l-i_{l-1}$, $j_1=i_1$.
The Lie algebra $\ofrak_{l+1}$ is the conjugacy limit of $\ofrak_l$ via the sequence $b_n^{(l+1)}$ for $n\to\infty$ and therefore has the form
\begin{equation*}
\ofrak_{l+1} = \pfrak \left(\begin{matrix}
\ofrak(j_1) & & & & \\
\rr^{j_2\times j_1} & \ofrak(j_2) &&&\\
\vdots & \vdots & \ddots &&\\
\rr^{j_{l+1}\times j_1} & \rr^{j_{l+1}\times j_2} & \cdots & \ofrak(j_{l+1}) &0 \\
\rr^{(m-i_{l+1})\times j_1} & \rr^{(m-i_{l+1})\times j_2} & \cdots & \rr^{(m-i_{l+1})\times j_{l+1}} & \ofrak(m-i_{l+1})
\end{matrix}\right)\,.
\end{equation*}
We define a subalgebra of both $\ofrak_{l}$ and $\ofrak_{l+1}$ as
\begin{equation*}
\lfrak_{l+1} = \pfrak \left(\begin{matrix}
\ofrak(j_1) & & & & \\
\rr^{j_2\times j_1} & \ofrak(j_2) &&&\\
\vdots & \vdots & \ddots &&\\
\rr^{j_{l+1}\times j_1} & \rr^{j_{l+1}\times j_2} & \cdots & \ofrak(j_{l+1}) &0 \\
0 & 0 & \cdots & 0 & \ofrak(m-i_{l+1})
\end{matrix}\right)\,.
\end{equation*}
Restricting the isomorphism $\sigma_l$ to the preimage of $\lfrak_{l+1}$ yields an isomorphism onto $\lfrak_{l+1}$.
We contract $\hfrak_l$ along
\begin{equation*}
\tfrak_l = \{X\in\hfrak_l\,|\, \Ad_{b_n^{(l+1)}} X = X\,\forall n\}\,,
\end{equation*}
which we write as the Lie algebra $\hfrak_{l+1}$ with commutator $[\cdot,\cdot]_{l+1}$.
In the image of $\sigma_l$, the map $\Ad_{b_n^{(l+1)}}$ leaves exactly the elements of $\lfrak_{l+1}$ invariant, such that the preimage of $\lfrak_{l+1}$ is $\sigma_l^{-1}(\lfrak_{l+1}) = \tfrak_l$.
We define a map $\sigma_{l+1}:\hfrak_{l+1}\to \ofrak_{l+1}$ as follows: set $\sigma_{l+1}|_{\tfrak_l} = \sigma_l|_{\tfrak_l}$ and for $X\in \rr^{(m-i_{l+1})\times i_{l}}$, $Y\in \rr^{(m-i_{l+1})\times j_{l+1}}$:
\begin{equation}\label{EqMatricesLatterType}
\sigma_{l+1}\left(\Pp\left(\begin{matrix}
0 &&  \\
0 & 0 & -Y^T\\
X & Y & 0
\end{matrix}\right)\right):= \Pp\left(\begin{matrix}
0 &&\\
0 & 0 & 0\\
X & Y & 0
\end{matrix}\right)\,,
\end{equation}
which is a vector space isomorphism with inverse similar to \Cref{EqSigma1Inverse}.
The complement $\lfrak_{l+1}^c$ of $\lfrak_{l+1}$ in $\ofrak_{l+1}$ is formed by matrices of the right-hand side form of \Cref{EqMatricesLatterType}.
Before the conjugacy limit of $\ofrak_l$ via $b_n^{(l+1)}$, note that matrix commutators of two matrices in $\lfrak_{l+1}\subset \ofrak_l$ are invariant under $\Ad_{b_n^{(l+1)}}$ for all $n$.
Considering the limit algebra $\ofrak_{l+1}$, the matrix commutator of two matrices in $\lfrak_{l+1}^c$ equates to zero, and those of $\ofrak_{l+1}$- with $\lfrak_{l+1}^c$-elements are in $\lfrak_{l+1}^c$.
Comparing with the construction of $\hfrak_{l+1}$, $\sigma_{l+1}$ is thus a Lie algebra morphism, such that $\hfrak_{l+1}\cong \ofrak_{l+1}$.

To summarize the inductive argument, we have the commutative diagram
\begin{equation*}
\begin{tikzcd}
\pofrak(m)\arrow[equal]{d}\arrow{r}{\begin{subarray}{c} \textrm{contr.} \\ \textrm{along }\zfrak_1 \end{subarray}} &[3em]  \hfrak_1\arrow{r}{\begin{subarray}{c} \textrm{contr.} \\ \textrm{along }\tfrak_1 \end{subarray}}\arrow[swap, "\sigma_1", "\sim"' labl]{d} &[3em]  \hfrak_2\arrow{r}{\begin{subarray}{c} \textrm{contr.} \\ \textrm{along }\tfrak_2 \end{subarray}}\arrow[swap, "\sigma_2", "\sim"' labl]{d} &[3em] \cdots\arrow{r}{\begin{subarray}{c} \textrm{contr.} \\ \textrm{along }\tfrak_{s-1} \end{subarray}} &[3em] \hfrak_s\arrow[swap, "\sigma_s", "\sim"' labl]{d}\\
\pofrak(m) \arrow[swap]{r}{\lim_{n\to\infty}\Ad_{b_n^{(1)}}} & \ofrak_1\arrow[swap]{r}{\lim_{n\to\infty}\Ad_{b_n^{(2)}}} & \ofrak_2\arrow[swap]{r}{\lim_{n\to\infty}\Ad_{b_n^{(3)}}} & \cdots\arrow[swap]{r}{\lim_{n\to\infty}\Ad_{b_n^{(s)}}} & \ofrak_s = \ofrak'
\end{tikzcd}\,,
\end{equation*}
where all vertical arrows are isomorphisms of Lie algebras.

Conjugation of both $\hfrak_s$ and $\ofrak_s = \ofrak'$ by elements in the maximal compact subgroup $K<\PGL_m\rr$ is an isomorphism.
Thus, $\hfrak_s$ and $\ofrak_s$ are isomorphic for all sequences $b_n\in \PGL_m\rr = KB\,\PO(p,q)$, since sequences in $B$ are up to left multiplication with sequences in $K$ and coordinate permutations of the form specified at the beginning of the proof, see e.g.~the proof of Thm.~1.1 in~\cite{cooper2018limits}.
All conjugacy limits of $\pofrak(m)$ via sequences in $\PGL_m\rr$ are therefore isomorphic to the composition of a finite number of contractions.

Lie algebras of the indefinite orthogonal groups are of the form
\begin{equation*}
\pofrak(p,q) = \{X\in\pgl_m\rr\,|\, \sigma(\theta(X)) = X\}\,,
\end{equation*}
where $\sigma(Y) = J Y J^{-1}$ with $J=-1_{p\times p}\oplus 1_{q\times q}$.
There exists a set of generators of $\pofrak(m)$, which after suitable sign changes generate $\pofrak(p,q)$; similar sign changes occur for their commutators.
Still, for $Y\in\pofrak(p,q)$ the off-diagonal elements $Y_{ij}$, $i\neq j$, are uniquely determined by the transposed elements $Y_{ji}$.
Hence, analogously to the $\sigma_l$ for conjugacy limits of $\pofrak(m)$, isomorphisms $\sigma_l'$ can be constructed, which respect the necessary sign changes as we exchange $\pofrak(m)$ for $\pofrak(p,q)$.
The proof of the claim for $\pofrak(m)$ hence applies also to the more general $\pofrak(p,q)$.
\end{proof}

\section{Projective representations and conjugacy limits}\label{AppendixConjugacyLimitsInAmbientRep}

\subsection{A commutative diagram for conjugacy limits}\label{AppendixCommutativeDiagramReps}

Let $G$ be a Lie group and $H<G$ a closed Lie subgroup. 
While $G=\PGL_{m}\rr$ for $m=5$ suffices for our purposes, we keep the statements of this appendix general.
We denote the space of invertible linear operators on the Hilbert space $\Hcal$ by $\GL(\Hcal)$ and let $U:G\to \GL(\Hcal)$ be a projective complex representation of $G$ on $\Hcal$.
We have the following lemma.
\vs

\begin{lem}\label{LemmaRepresentationLimits}
For $L$ a conjugacy limit of $H$ in $G$ and $U$ a projective complex representation of $G$, the diagram
\begin{equation*}
\begin{tikzcd}
\lim_{n\to\infty}\Ad_{b_n}: & H\arrow{r}\arrow{d}{U} & L\arrow{d}{[U]}\\
\left[ \lim_{n\to\infty}\Ad_{U(b_n)} \right]: & U(H) \arrow{r} & \left[U(L)\right]
\end{tikzcd} 
\end{equation*}
commutes, where $[U(L)]$ denotes equivalence classes modulo $U(1)$ prefactors.
\end{lem}
\vs

By virtue of \Cref{LemmaRepresentationLimits}, limits of projective representations within an ambient projective representation are the same as projective representations of limits up to limits of multipliers.
For its proof we first show two propositions.
\vs

\begin{prop}\label{PropConjugacyRepresentation1}
Let $L$ be a conjugacy limit of $H$ in $G$. 
Then $U(L)$ is up to multiplication by the multiplier of $U$ a conjugacy limit of $U(H)$ in $U(G)$.
\end{prop}

\begin{proof}
Let $b_n\in G$ be a sequence, such that $b_n H b_n^{-1}$ converges geometrically to the closed subgroup $L<G$. 
Consider $g\in L$ and $\tilde{h}_n=b_n h_n b_n^{-1}\in b_n H b_n^{-1}$ with $\tilde{h}_n\to g$ for $n\to \infty$.
Then $[U(\tilde{h}_n)] =  [U(b_n) U (h_n) U(b_n)^{-1}]$ for all $n$ and $U(\tilde{h}_n)\to U(g)\in U(L)$ by continuity of $U$. 
Every accumulation point of $\tilde{h}_n$ lies in $L$, so every accumulation point of $U(\tilde{h}_n)$ lies in $U(L)$. 
Indeed, $U(L)$ is a conjugacy limit of $U(H)$ in $U(G)$ up to (limits of) multipliers of $U$.
\end{proof}

Chabauty topology is the subspace topology on the set $\mathcal{C}(G)$ of closed subgroups of $G$. 
$\mathcal{C}(G)$ with the Chabauty topology is a compact Hausdorff topological space, see e.g.~\cite{leitner2016limits,cooper2018limits,trettel2019families}. 
In general, a subgroup $L<G$ is a conjugacy limit of $H<G$ via a sequence $b_n\in G$, if and only if $H_n=b_n H b_n^{-1}\to L$ in the Chabauty topology~\cite{cooper2018limits}. 
This leads to the following proposition.
\vs

\begin{prop}\label{PropConjugacyRepresentation2}
Every conjugacy limit of $U(H)$ in $U(G)$ is of the form $U(L)$ for a unique conjugacy limit $L$ of $H$ in $G$, up to limits of multipliers of $U$.
\end{prop}

\begin{proof}
Denote the multiplier of $U$ by $\omega$.
Let $U'$ be a conjugacy limit of $U(H)$ in $U(G)$. 
Then, any $U(g)\in U'$ is the limit of some sequence $U(h_n')\in U(b_n) U(H) U(b_n)^{-1}$ with 
\begin{equation*}
U(h_n')=U(b_n)U(h_n)U(b_n)^{-1} = \omega(b_n,h_n)\omega(b_n h_n,b_n^{-1})U(\tilde{h}_n)
\end{equation*}
for $b_n\in G, h_n\in H$, $\tilde{h}_n:= b_n h_n b_n^{-1}$.
Assume that $U$ is faithful. 
Then, given that $U(h_n')\to U(g)$ in $U(G)$, we obtain $\tilde{h}_n\to g$. 
The same argument applies to accumulation points of $U(h_n')\in U(G)$, which come from accumulation points of the sequence $\tilde{h}_n\in G$. 
Thus, $b_n H b_n^{-1}$ converges to a conjugacy limit $L < G$. 
The set $\mathcal{C}(G)$ of closed subgroups equipped with the Chabauty topology being Hausdorff, limits are unique. 
Therefore, $[U'] = [U(L)]$ for a limit $L$ of $H$ in $G$.

The case of a non-faithful projective representation $U$ remains. 
Let $\ker(U)\neq \{1\}$ be the non-trivial kernel of $U$, which is a normal closed subgroup of $G$, such that we obtain the exact sequence of groups
\begin{equation*}
1 \to \ker(U)< G \to G /\ker(U)\to 1\, .
\end{equation*}
A limit $U'$ of $U(b_n) U(H)U(b_n)^{-1} = \omega(b_n,h)\omega(b_n h,b_n^{-1}) U(b_n H b_n^{-1}) < U(G)$ for a sequence $b_n\in G$ induces the same limit $[U']$ within $[U(G/\ker(U))]$, since conjugation preserves $\ker(U)$. 
$U$ acting on $G/\ker(U)$ is a faithful projective representation, such that $[U'] = [U(\ker(U)\cdot L )] = [U(L)]$ for a limit $L$ of $H<G$. 
Conjugacy limits of subgroups of $U(G)$ and $U(G/\ker(U))$ coincide trivially. 
Indeed, for a general representation $U$ we find $[U'] = [U(L)]$ for a limit $L$ of $H<G$.
\end{proof}

The proof of \Cref{LemmaRepresentationLimits} is now immediate.

\begin{proof}[Proof of \Cref{LemmaRepresentationLimits}]
The implied equality of morphism concatenations follows from Propositions \ref{PropConjugacyRepresentation1} \mbox{and \ref{PropConjugacyRepresentation2}} and their proofs.
\end{proof}

\subsection{Projective representations and projective quantum fields}\label{AppendixAspectsProjectivity}
The following proposition shows that requiring $\rho$ non-projective for projective quantum fields as in \Cref{DefFieldOperator} is no restriction.
\vs

\begin{prop}\label{PropRhoHasTrivialMultiplier}
Assume that for a projective quantum field $(U,\rho,\{[\hat{\Ocal}([x])]\})$ the representation $\rho$ is projective instead of non-projective. 
Then $\rho$ comes with trivial multiplier independent from the multiplier of $U$.
\end{prop}

\begin{proof}
Assume $U$ and $\rho$ have multipliers $\omega_U$ and $\omega_\rho$, respectively.
The generalized unitary transformation behavior \eqref{EqGlobalCovariance} of representatives $\hat{\Ocal}([x])$ of the field operators $[\hat{\Ocal}([x])]$, $[x]\in \RP^4$, implies for ${[g],[h]\in\PGL_5\rr}$:
\begin{align}\label{EqPropRhoTrivialMultiplierProof}
\sum_{\beta,\gamma}\rho_{\alpha\beta}([h^{-1}])\rho_{\beta\gamma}([g^{-1}]) \hat{\Ocal}_\gamma([gh\cdot x]) =&\;  \sum_\beta\rho_{\alpha\beta}([h^{-1}])U([g])\hat{\Ocal}_\beta([h\cdot x]) U^\dagger([g]) \nonumber\\
=&\; U([g])U([h]) \hat{\Ocal}_\alpha([x]) U^\dagger([h]) U^\dagger([g])\nonumber\\
=&\; \omega_U([g],[h])^2 U([gh]) \hat{\Ocal}_\alpha([x]) U^\dagger([gh])\nonumber\\
=&\; U([gh]) \hat{\Ocal}_\alpha([x]) U^\dagger([gh])\nonumber\\
=&\; \sum_\beta \rho_{\alpha\beta}([h^{-1} g^{-1}]) \hat{\Ocal}_\beta([gh\cdot x])\,,
\end{align}
where we employed that $\omega_U([g],[h])^2 = 1$ for all $[g],[h]\in\PGL_5\rr$ by \Cref{PropCohomologyGeneratorsPGL5}.
By \Cref{EqPropRhoTrivialMultiplierProof}:
\begin{equation*}
\hat{\Ocal}([gh\cdot x]) = \rho([g]) \rho([h]) \rho([h^{-1}g^{-1}]) \hat{\Ocal}([gh\cdot x]) = \omega_\rho([g],[h]) \hat{\Ocal}([gh\cdot x])\,,
\end{equation*}
i.e., $\omega_\rho([g],[h])=+1$.
\end{proof}


\section{Dependence of projective quantum fields on the ambient geometry}\label{AppendixAmbientGeometryChoice}

The construction of projective quantum fields depends on the non-unique choice of an ambient projective geometry, which can incorporate four-dimensional geometries as subgeometries, here taken to be $(\RP^4,\PGL_5\rr)$. 
We show that this is at least partly without loss of generality.
\vs

\begin{prop}\label{PropDeformationsLimitsLargeAmbientGeometry}
Let $\Gg=(\Xx,\Oo)<(\RP^4,\PGL_5\rr)$ be a four-dimensional geometry. 
If $\Gg$ is viewed as a subgeometry of $(\RP^{m-1},\PGL_{m})$, $m\geq 5$, all deformations and limits preserving the subgeometry $(\RP^4,\PGL_5\rr)<(\RP^{m-1},\PGL_{m})$, which in turn contains $\Gg$, arise up to conjugation within $(\RP^{m-1},\PGL_{m})$ from deformations and limits of $\Gg$ in $(\RP^4,\PGL_5\rr)$.
\end{prop}

\begin{proof}
With $\Oo< \PGL_{m}\rr$ a subgroup, it is conjugate within $\PGL_{m}\rr$ for $m\geq 5$ to the canonical form
\begin{equation*}
 \Pp(((\Pp^{-1}\Oo)\cap \SL_5\rr)\times 1):= \Pp \left(\begin{matrix}
(\Pp^{-1}\Oo)\cap \SL_5\rr & \\
 & 1_{(m-5)\times (m-5)}
\end{matrix}\right)\,.
\end{equation*}
Let $[b_n]\in \PGL_{m}\rr$. 
We write
\begin{equation*}
[b_n] = \Pp\left(\begin{matrix}
A_n & B_n\\
C_n & D_n
\end{matrix}\right)\,,
\end{equation*}
with $A_n\in \GL_5\rr$ and all other submatrices of matching sizes.
Preservation of $(\RP^4,\allowbreak \PGL_5\rr)< (\RP^{m-1},\PGL_{m})$ under the adjoint action of $[b_n]$, where $\Gg < (\RP^4,\allowbreak \PGL_5\rr)< (\RP^{m-1},\allowbreak \PGL_{m})$, yields $B_n = 0$ and $C_n=0$.
Then, 
\begin{equation*}
[b_n] \Pp(((\Pp^{-1}\Oo)\cap \SL_5\rr)\times 1) [b_n^{-1}] = \Pp(A_n ((\Pp^{-1}\Oo)\cap \SL_5\rr) A_n^{-1} \times 1)\,.
\end{equation*}
Therefore, deformations and limits of $\Oo<\PGL_{m}\rr$ are up to conjugation within $\PGL_{m}\rr$ the same as within $\PGL_5\rr$, if the subgroup $\PGL_5\rr< \PGL_m\rr$ containing $\Oo$ is preserved. 
The claim for the related model spaces follows analogously.
\end{proof}

\begin{lem}
Consider a projective quantum field $\hat{\Ocal} = (U,\rho,\{[\hat{\Ocal}([x])]\,|\, [x]\in\RP^{m-1}\})$ for the ambient geometry $(\RP^{m-1},\PGL_{m}\rr)$, $m\geq 5$, i.e., $U:\PGL_{m}\rr\to \Ucal(\Hcal)$ is a projective unitary representation and $\rho$ is a finite-dimensional complex $\overline{\PGL_{m}\rr}$ representation, so that the generalized unitary transformation behavior \eqref{EqGlobalCovariance} holds on all $(\RP^{m-1},\PGL_m\rr)$.
Restricted to a four-dimensional geometry $\Gg = (\Xx,\Oo)<(\RP^{m-1},\allowbreak \PGL_{m}\rr)$, deformations and limits of $\hat{\Ocal}|_{(\Xx,\Oo)}$, which preserve the second embedding of $(\Xx,\Oo) < (\RP^4,\PGL_5\rr) < (\RP^{m-1},\PGL_{m}\rr)$, agree up to conjugation in $\PGL_{m}\rr$ with those of $\hat{\Ocal}$ for the ambient geometry $(\RP^4,\PGL_5\rr)$.
\end{lem}

\begin{proof}
The claim is a direct consequence of \Cref{PropDeformationsLimitsLargeAmbientGeometry} together with projective quantum field properties.
\end{proof}

This does not include ambient geometries which are not real projective geometries. 
Yet, in three dimensions one can mathematically heuristically argue in favor of the projective geometry setting.
Namely, all eight Thurston geometries included in Thurston's geometrization program (almost) admit a representation in real projective geometry~\cite{molnar1997projective, thiel1997einheitliche}. 

\end{appendices}

\bibliographystyle{apsrev4-1_custom}
\bibliography{literature}

\end{document}